\documentclass[final,3p]{elsarticle}
 \usepackage[latin1]{inputenc}
 \usepackage{graphics}
 \usepackage{graphicx}
 \usepackage{epsfig}
\usepackage{amssymb}
 \usepackage{amsthm}
 \usepackage{lineno}
 \usepackage{amsmath}
   \numberwithin{equation}{section}
\usepackage{mathrsfs}

\newtheorem{thm}{Theorem}[section]

\newtheorem{lem}[thm]{Lemma}

\newtheorem{defn}[thm]{Definition}

\biboptions{sort&compress,square}
\begin{document}
\begin{frontmatter}
\author[rvt1]{Jian Wang}
\ead{wangj@tute.edu.cn}
\author[rvt2]{Yong Wang\corref{cor2}}
\ead{wangy581@nenu.edu.cn}

\cortext[cor2]{Corresponding author.}
\address[rvt1]{School of Science, Tianjin University of Technology and Education, Tianjin, 300222, P.R.China}
\address[rvt2]{School of Mathematics and Statistics, Northeast Normal University,
Changchun, 130024, P.R.China}

\title{Spectral Torsion for Nonminimal de-Rham Hodge Operator}
\begin{abstract}
In this paper, we investigate  some new spectral  torsion  which is the extension of spectral  torsion  for Dirac operators, and compute the spectral torsion  associated with nonminimal de Rham-Hodge operators on manifolds with (or without)  boundary.
The novelty of this paper lies in the calculation related to the noncommutative residue for the inverse of the Laplacion with  non-scalar principle symbols.
\end{abstract}
\begin{keyword}
Nonminimal de-Rham Hodge operator; Spectral  torsion; Noncommutative residue.
\MSC[2000] 53G20, 53A30, 46L87
\end{keyword}
\end{frontmatter}
\section{Introduction}
\label{1}
The noncommutative residue found in \cite{Gu,Wo1} plays a prominent role in noncommutative geometry.
For one dimensional manifolds, the noncommutative residue was discovered by Adler \cite{MA}
 in connection with geometric aspects of nonlinear partial differential equations.
Various extensions of the definition of Wodzicki residue have been considered
by many authors. In noncompact environments, Nicola\cite{FN} considered
traces on an algebra of pseudodifferential operators in $\mathbb{R}^{n}$. Battisti
and Coriasco \cite{BC1,BC2} extended the concept of Wodzicki residue to operator
algebras on the class of the manifolds with a finite number of cylindrical
ends.
In \cite{Co1}, Connes used the noncommutative residue to derive a conformal four dimensional Polyakov action analogy. Moreover, in \cite{Co2},
Connes made a challenging observation that the noncommutative residue of the square of the inverse of the Dirac operator was proportional to
the Einstein-Hilbert action, then Kastler\cite{Ka}, Kalau and Walze\cite{KW} gave a brute-force proof of this theorem, which we call the Kastler-Kalau-Walze theorem.
In \cite{FGLS}, Fedosov et al. defined a noncommutative residue on Boutet de Monvel's algebra and proved that it was a
unique continuous trace. In \cite{S}, Schrohe gave the relation between the Dixmier trace and the noncommutative residue for
manifolds with boundary.  For an oriented spin manifold  $M$ with boundary $\partial M$,  by the composition formula in Boutet de Monvel's algebra and the definition of $ \widetilde{{\rm Wres}}$ \cite{Wa1},  $\widetilde{{\rm Wres}}[(\pi^+D^{-1})^2]$ should be the sum of two terms from
interior and boundary of $M$, where $\pi^+D^{-1}$ is an element in Boutet de Monvel's algebra  \cite{Wa1}.
 For lower-dimension spin manifolds with boundary and the associated Dirac operators, Wang computed the lower dimensional volume and got
 a Kastler-Kalau-Walze type theorem in \cite{Wa4,Wa3,Wa2}. In\cite{GBF}, Gilkey, Branson and Fulling obtained a formula about
  heat kernel expansion coefficients of nonminimal operators. In \cite{WW}, we considered  the non-commutative residue of nonminimal operators
   and got the Kastler-Kalau-Walze type theorems for  nonminimal operators. In \cite{WW1}, we prove various Kastler-Kalau-Walze type theorems associated with  nonminimal de Rham-Hodge operators on compact manifolds with  boundary.

 In \cite{DL}, Dabrowski etc. obtained the Einstein
tensor (or, more precisely, its contravariant version)  using the Clifford representation of one-forms as $0$-order differential operators,
and they demonstrated that the noncommutative residue density recovered
the tensors $g$ and $G:=Ric-\frac{1}{2}R(g)g$ as certain bilinear functionals of vector fields on a manifold $M$, while their
dual tensors are recovered as a density of bilinear functionals of differential one-forms on $M$.
For Riemannian manifold $M$ of even dimension $n = 2m$ equipped
with a metric tensor $g$ and the (scalar) Laplacian $\Delta$, a localised functional in $C^{\infty}(M)$ can be defined by
  \begin{equation}
  {\rm  Wres}(f\Delta^{-m+1})=\frac{n-2}{12}\upsilon_{n-1}\int_{M}fR(g){\rm vol}_{g},
\end{equation}
where $f\in C^{\infty}(M)$, $ R = R(g)$ is the scalar curvature, that is the $g$-trace $R=g^{jk}R_{jk }$ of
the Ricci tensor with components $R_{jk}$  in local coordinates,  $g^{jk}$ are the raised components
of the metric $g$.
Let $u$, $v$,$w$ with the components with
respect to local coordinates $u_{a}$,$v_{a}$ and $w_{a}$, respectively, be  differential forms represented in
such a way as endomorphisms (matrices) $c(u) $, $c(v) $ and $c(w) $ on the spinor bundle.
The torsion functional, as defined in \cite{DL2}, assigns to a triple of one-forms $(u,v,w)$:
\begin{equation}
	\mathscr{T}_D(u,v,w)= \mathrm{Wres}(c(u)c(v)c(w)D D^{-2m}), \qquad u,v,w\in \Omega^1_D.
\end{equation}
Finally, the scalar curvature functional is defined by $\mathscr{R}_D(f)=Wres(fD^{-2m+2})$ for $f\in \mathcal{A}$.
In \cite{WW2},  we give some new spectral functionals which is the extension of spectral functionals to the noncommutative realm with torsion,
 and we relate them to  the noncommutative residue for manifolds with boundary.
For a finitely summable regular spectral triple, we recover two forms, torsion of the linear connection and four forms by the noncommutative residue and perturbed de-Rham Hodge operators, and provide an explicit computation of generalized spectral forms associated with the perturbed de-Rham Hodge Dirac triple  in \cite{WW3}.

 Motivated by twisted spectral triple,  the spectral torsion for Dirac operators  and the non-commutative residue of nonminimal operators,
the purpose of this paper is to generalize the results of spectral  torsion  for Dirac operators and get
the spectral torsion  associated with nonminimal de-Rham Hodge operators.
In our spectral torsion of the  nonminimal de-Rham Hodge operators, the  nonminimal Laplacion naturally appears. Usually,
the Laplacion  related to the computations of the noncommutative residue has the  scalar principle symbol.
In this paper, our computations are related to the Laplacion with  non-scalar principle symbols, which makes our calculation more complex.
  This paper is organized as follows: In Section 2, we recall  preliminaries on the spectral torsion  for  nonminimal de Rham-Hodge operator.
   In Section 3, for four dimensional compact manifolds without boundary and  the associated nonminimal de-Rham Hodge operators,
we compute the spectral torsion on compact manifolds.
 In Section 4, we compute the spectral torsion  associated with nonminimal de Rham-Hodge operators $\widetilde{D}=a_0d+b_0\delta$ on four-dimensional Riemannian manifolds with boundary, where $a_0, b_0$ are constants.

\section{Preliminaries on the spectral torsion  for  nonminimal de Rham-Hodge operators}

\subsection{Nonminimal de-Rham Hodge operators}
Firstly, we introduce some notations about Clifford action and nonminimal de Rham-Hodge operators. Let $M$ be an $n$-dimensional ($n\geq 3$) oriented compact Riemannian manifold with a Riemannian metric $g^{TM}$.

Let $\nabla$ denote the Levi-civita connection about $g^M$.
 In the local coordinates $\{x_i; 1\leq i\leq n\}$ and the fixed orthonormal frame $\{\widetilde{e_1},\cdots,\widetilde{e_n}\}$,
 the connection matrix $(\omega_{s,t})$ is defined by
  \begin{equation*}
\nabla(\widetilde{e_1},\cdots,\widetilde{e_n})= (\widetilde{e_1},\cdots,\widetilde{e_n})(\omega_{s,t}).
\end{equation*}
Let $\widehat{c}(e_j)=\epsilon (e_j^* )+\iota
(e_j);~~
c(e_j)=\epsilon (e_j^* )-\iota (e_j),$
 where $e_j^*=g^{TM}(e_j,\cdot)$ and $\epsilon (e_j^*)$,~$\iota (e_j)$ be the exterior and interior multiplications respectively,
which satisfies
\begin{align}
\label{ali1}
&\widehat{c}(e_i)\widehat{c}(e_j)+\widehat{c}(e_j)\widehat{c}(e_i)=2g^{TM}(e_i,e_j);~~\nonumber\\
&c(e_i)c(e_j)+c(e_j)c(e_i)=-2g^{TM}(e_i,e_j);~~\nonumber\\
&c(e_i)\widehat{c}(e_j)+\widehat{c}(e_j)c(e_i)=0.
\end{align}
Next, we recall the nonminimal de Rham-Hodge operator. Following \cite{WW1}, the de Rham derivative $\text{d}$ is a differential operator on $C^{\infty}(M,\wedge^*T^*M)$ and the de Rham coderivative $\delta=\text{d}^*$. Denote
by  $\text{d}+\delta:~\wedge^*(T^*M)\rightarrow \wedge^*(T^*M)$  the signature operator.
By \cite{WW1}, we have
  \begin{equation*}
\text{d}+\delta=\sum^n_{i=1}c(e_i)\Big[e_i+\frac{1}{4}\sum_{s,t}\omega_{s,t}
(e_i)\big[\hat{c}(e_s)\hat{c}(e_t)-c(e_s)c(e_t)\big]\Big].
\end{equation*}
Let $\widetilde{c}(e_j)=a_0\epsilon(e_j^*)-b_0\iota(e_j)$, where $a_0, b_0$ are constants and $a_0b_0\neq 0$,  we define the nonminimal de Rham-Hodge operators  as
\begin{align}
\widetilde{D}=&a_0d+b_0\delta+\sqrt{-1}c(X)
=\sum^n_{i=1}\widetilde{c}(e_i)\bigg[e_i+\frac{1}{4}\sum_{s,t}\omega_{s,t}
(e_i)[\widehat{c}(e_s)\widehat{c}(e_t)
-c(e_s)c(e_t)]\bigg]+\sqrt{-1}c(X),
\end{align}
and $\widetilde{D}^{*}=b_0\text{d}+a_0\delta+\sqrt{-1}c(X)$ is the adjoint operator of $\widetilde{D}$.

\subsection{Noncommutative residue for manifold with boundary}
 In this section we consider an $n$-dimensional oriented Riemannian manifold $(M, g^{M})$ equipped
with some spin structure. Let $M$ be an $n$-dimensional compact oriented manifold with boundary $\partial M$.
 We assume that the metric $g^{M}$ on $M$ has
the following form near the boundary
 \begin{equation}
 g^{M}=\frac{1}{h(x_{n})}g^{\partial M}+\text {d}x _{n}^{2} ,
\end{equation}
where $g^{\partial M}$ is the metric on $\partial M$ and  $h(x_n)\in
C^{\infty}([0,1))=\{\widetilde{h}|_{[0,1)}|\widetilde{h}\in
C^{\infty}((-\varepsilon,1))\}$ for some $\varepsilon>0$ and
satisfies $h(x_n)>0,~h(0)=1$ where $x_n$ denotes the normal
directional coordinate.
 Let $U\subset
M$ be a collar neighborhood of $\partial M$ which is diffeomorphic $\partial M\times [0,1)$. By the definition of $h(x_n)\in C^{\infty}([0,1))$
and $h(x_n)>0$, there exists $\tilde{h}\in C^{\infty}((-\varepsilon,1))$ such that $\tilde{h}|_{[0,1)}=h$ and $\tilde{h}>0$ for some
sufficiently small $\varepsilon>0$. Then there exists a metric $\hat{g}$ on $\hat{M}=M\bigcup_{\partial M}\big(\partial M\times
(-\varepsilon,0]\big)$ which has the form on $U\bigcup_{\partial M}\big(\partial M\times (-\varepsilon,0 ]\big)$
 \begin{equation*}
\hat{g}=\frac{1}{\tilde{h}(x_{n})}g^{\partial M}+\text {d}x _{n}^{2} ,
\end{equation*}
such that $\hat{g}|_{M}=g$. We fix a metric $\hat{g}$ on the $\hat{M}$ such that $\hat{g}|_{M}=g$.

To define the lower dimensional volume, some basic facts and formulae about Boutet de Monvel's calculus which can be found  in Sec.2
in \cite{Wa1} are needed. Let $F:L^2(\mathbb{R}_t)\rightarrow L^2(\mathbb{R}_v)$
  denote the Fourier transformation and
$\Phi(\overline{\mathbb{R}^+}) =r^+\Phi(\mathbb{R})$ (similarly define $\Phi(\overline{\mathbb{R}^-}$)), where $\Phi(\mathbb{R})$
denotes the Schwartz space and
  \begin{equation*}
r^{+}:C^\infty (\mathbb{R})\rightarrow C^\infty (\overline{\mathbb{R}^+});~ f\rightarrow f|\overline{\mathbb{R}^+};~
 \overline{\mathbb{R}^+}=\{x\geq0;x\in \mathbb{R}\}.
\end{equation*}
Fourier transformation, which, for $ u\in L^{2}(\mathbb{R}_{t})\cap L^{1}(\mathbb{R}_{t})$, is given by
\begin{equation*}
F(u)( \xi )=\int e^{-i\xi t}u(t)\text{d}t.
  \end{equation*}
We define $H^+=F(\Phi(\overline{\mathbb{R}^+}));~ H^-_0=F(\Phi(\overline{\mathbb{R}^-}))$ which are orthogonal to each other. We have the following
 property: $h\in H^+~(H^-_0)$ iff $h\in C^\infty(\mathbb{R})$ and it admits an analytic extension to the lower (upper resp.) complex
half-plane $\{{\rm Im}\xi<0\}~(\{{\rm Im}\xi>0\}$ resp.) such that for all nonnegative integer $l$,
 \begin{equation*}
\frac{\text{d}^{l}h}{ \text{d}\xi^l}(\xi)\sim\sum^{\infty}_{k=1}\frac{\text{d}^l}{\text{d}\xi^l}(\frac{c_k}{\xi^k})
\end{equation*}
as $|\xi|\rightarrow +\infty,{\rm Im}\xi\leq0~({\rm Im}\xi\geq0)$.

 Let $H'$ be the space of all polynomials and $H^-=H^-_0\bigoplus H';~H=H^+\bigoplus H^-.$ Denote by $\pi^+~(\pi^-)$ respectively the
 projection on $H^+~(H^-)$. For calculations, we take $H=\widetilde H=\{$rational functions having no poles on the real axis$\}$ ($\widetilde{H}$
 is a dense set in the topology of $H$). Then on $\tilde{H}$,
 \begin{equation*}
\pi^+h(\xi_0)=\frac{1}{2\pi i}\lim_{u\rightarrow 0^{-}}\int_{\Gamma^+}\frac{h(\xi)}{\xi_0+iu-\xi}\text {d}\xi,
\end{equation*}
where $\Gamma^+$ is a Jordan close curve included ${\rm Im}\xi>0$ surrounding all the singularities of $h$ in the upper half-plane and
$\xi_0\in \mathbb{R}$. Similarly, define $\pi'$ on $\widetilde{H}$,
 \begin{equation*}
\pi'h=\frac{1}{2\pi}\int_{\Gamma^+}h(\xi)\text {d}\xi.
\end{equation*}
So, $\pi'(H^-)=0$. For $h\in H\bigcap L^1(R)$, $\pi'h=\frac{1}{2\pi}\int_{R}h(v)dv$ and for $h\in H^+\bigcap L^1(R)$, $\pi'h=0$.

Denoting by $\mathcal{B}$ the Boutet de Monvel's algebra, we recall the main theorem in \cite{FGLS}.
\begin{thm}\label{th:32}\cite{FGLS}
Let $X$ be a manifold with boundary $\partial X$ and
$E(F)$ are vector bundles over $X~(\partial X$ resp.), ${\rm dim}X=n\geq3$,
 $A=\left(\begin{array}{lcr}\pi^+P+G &   K \\
T &  S    \end{array}\right)$ $\in \mathcal{B}$ , and denote by $p$, $b$ and $s$ the local symbols of $P,G$ and $S$ respectively.
 Let ${\rm{tr}}_E ({\rm{tr}}_F$ ) be a trace on $E(F$ resp.), define:
 \begin{eqnarray}
{\rm{\widetilde{Wres}}}(A)&=&\int_X\int_{\bf S}{\rm{Tr}}_E\left[p_{-n}(x,\xi)\right]\sigma(\xi)\text {d}x \nonumber\\
&&+2\pi\int_ {\partial X}\int_{\bf S'}\left\{{\rm Tr}_E\left[(b_{-n})(x',\xi')\right]+{\rm{tr}}
_F\left[s_{1-n}(x',\xi')\right]\right\}\sigma(\xi')\text{d}x',
\end{eqnarray}
where ${\bf S}~({\bf S}')$ be the unit sphere about
$\xi~(\xi' $ resp.) and $\sigma(\xi)~(\sigma(\xi'))$ be the corresponding
canonical $n-1~(n-2$ resp.) volume form.
Then~~ a) ${\rm \widetilde{Wres}}([A,B])=0 $, for any
$A,B\in\mathcal{B}$;~~ b) It is a unique continuous trace on
$\mathcal{B}/\mathcal{B}^{-\infty}$.
\end{thm}
Let $p_{1},p_{2}$ be nonnegative integers and $p_{1}+p_{2}\leq n$. Then by Sec 2.1 of \cite{Wa3},  we have
\begin{defn} Lower-dimensional volumes of Riemannian manifolds with boundary  are defined by
   \begin{equation}\label{}
  {\rm Vol}^{(p_1,p_2)}_nM:=\widetilde{{\rm Wres}}[\pi^+\widetilde{D}^{-p_1}\circ\pi^+(\widetilde{D}^{*})^{-p_2}].
\end{equation}
\end{defn}

Let $\sigma_{l}(A)$ be the $l$-order symbols of an operator $A$ and
 $\sigma^+_{r}=\sigma^+_{r}(\xi)=\pi^+\sigma_{r}(x',0,\xi',\xi_{n})$.
  For $n$ dimensional Riemannian manifolds with boundary,
 an application of (2.1.4) in \cite{Wa1} shows that
\begin{equation}
\widetilde{{\rm Wres}}[\pi^+\widetilde{D}^{-p_1}\circ\pi^+(\widetilde{D}^{*})^{-p_2}]=\int_M\int_{|\xi|=1}{\rm
Tr}_{\Lambda^{*}(T^{*}M)}[\sigma_{-n}(\widetilde{D}^{-p_1}  (\widetilde{D}^{*})^{-p_2})]\sigma(\xi)\text{d}x+\int_{\partial
M}\Phi,
\end{equation}
where
 \begin{eqnarray}
\Phi&=&\int_{|\xi'|=1}\int^{+\infty}_{-\infty}\sum^{\infty}_{j, k=0}
\sum\frac{(-i)^{|\alpha|+j+k+1}}{\alpha!(j+k+1)!}
 {\rm Tr}_{\Lambda^{*}(T^{*}M)}
\Big[\partial^j_{x_n}\partial^\alpha_{\xi'}\partial^k_{\xi_n}
\sigma^+_{r}(\widetilde{D}^{-p_1})(x',0,\xi',\xi_n)\nonumber\\
&&\times\partial^\alpha_{x'}\partial^{j+1}_{\xi_n}\partial^k_{x_n}\sigma_{l}
((\widetilde{D}^{*})^{-p_2})(x',0,\xi',\xi_n)\Big]d\xi_n\sigma(\xi')\text {d}x',
\end{eqnarray}
and the sum is taken over $r-k+|\alpha|+\ell-j-1=-n,r\leq-p_{1},\ell\leq-p_{2}$.

 \section{The spectral torsion for  nonminimal de-Rham Hodge operators}

 \subsection{Torsion functional for  nonminimal de Rham-Hodge operators}

Let $S^{*}M\subset T^{*}M$ denotes the co-sphere bundle on $M$
 and a pseudo-differential operator $P\in \Psi DO(E)$,
denote by $\sigma_{-n}^{P}$ the component of order
$-n$ of the complete symbol $\sigma^{P}= \sum_{i}\sigma_{i}^{P}$ of $P$
such that the equality
 \begin{equation}
{\rm  Wres}(P)= \int_{S^{*}M}\text{Tr}(\sigma_{-n}^{P}(x,\xi)){\rm  d}x {\rm  d}\xi.
\end{equation}
In \cite{Co1,Co2,Ka,KW}, it was shown that the noncommutative  residue ${\rm  Wres}(\Delta^{-n/2+1})$ of a generalized laplacian $\Delta$  on a complex vector bundle $E$ over a closed compact manifold $M$,
is the integral of the second coefficient of the heat kernel expansion of $\Delta$ up to a proportional factor.
In \cite{Co1}, the well-known Connes' trace theorem states the Dixmier trace of $-n$ order
pseudo-differential operator equals to its  noncommutative  residue up to a constant on
a closed $n-$dimensional manifold.
 Denote by $\Delta$ the Laplacian as above and $Tr_{\omega}$ the Dixmier trace,
then
 \begin{equation}
Tr_{\omega}((1+\Delta)^{-n/2})=\frac{1}{n}{\rm  Wres}((1+\Delta)^{-n/2})=\frac{1}{n}{\rm  dim}(E){\rm  Vol}(S^{n-1}){\rm  Vol}_{M}.
\end{equation}
\begin{defn}\cite{DL}
Let $c(u)=\sum_{r=1}^{n} u_{r}c(e_r), c(v)=\sum_{p=1}^{n} v_{p}c(e_{p}), c(w)=\sum_{q=1}^{n} w_{q}c(e_{q}),$ the trilinear Clifford multiplication by functional of differential one-forms $c(u), c(v)$, $ c(w)$
\begin{align}
\mathscr{T}_D(u,v,w)= \mathrm{Wres}\big(c(u)c(v)c(w)D D^{-2m}\big)
 \end{align}
  is called torsion functional for Dirac operator.
  \end{defn}
 Now for  the nonminimal de Rham-Hodge operators $\widetilde{D}=a_0\text{d}+b_0\delta+\sqrt{-1}c(X)$, we have
  \begin{defn}
By the trilinear Clifford multiplication by functional of differential one-forms $\widetilde{c}(u)$, $\widetilde{c}(v)$,  $ \widetilde{c}(w)$, the spectral torsion $\mathscr{T}_{\widetilde{D}}$ for $\widetilde{D}=a_0\text{d}+b_0\delta+\sqrt{-1}c(X),  \widetilde{D}^{*}=b_0\text{d}+a_0\delta+\sqrt{-1}c(X)$ defined by
\begin{align}
 	\mathscr{T}_{\widetilde{D}}\big(\widetilde{c}(u),\widetilde{c}(v),\widetilde{c}(w)\big)
 &=\mathrm{Wres}\big(\widetilde{c}(u) \widetilde{c}(v) \widetilde{c}(w) \widetilde{D}(\widetilde{D}^{*}\widetilde{D})^{-m}\big) \nonumber\\
  &=\int_{M}\int_{\|\xi\|=1} \operatorname{Tr}\left[\sigma_{-2 m}
\Big( \widetilde{c}(u) \widetilde{c}(v) \widetilde{c}(w)\widetilde{D}(\widetilde{D}^{*}\widetilde{D})^{-m}\Big)
\right](x, \xi)\sigma(\xi)\text{d}x.
 \end{align}
  \end{defn}
When $a_0= b_0=1$, we get the spectral torsion in  \cite{DL}.
 \subsection{The spectral torsion for  nonminimal de-Rham Hodge operators}
This section is designed to get the spectral torsion associated with nonminimal de Rham-Hodge operators on four-dimensional compact manifolds.
To simplify the complex calculations, we present the results of a four-dimensional manifold. By Definition 3.2,
the spectral torsion $\mathscr{T}_{\widetilde{D}}$ for $\widetilde{D}=a_0\text{d}+b_0\delta+\sqrt{-1}c(X),  \widetilde{D}^{*}=b_0\text{d}+a_0\delta+\sqrt{-1}c(X)$ on four-dimensional compact manifolds read
\begin{align}
 \mathscr{T}_{\widetilde{D}}\big(\widetilde{c}(u),\widetilde{c}(v),\widetilde{c}(w)\big)
 &=\mathrm{Wres}\big(\widetilde{c}(u) \widetilde{c}(v) \widetilde{c}(w) \widetilde{D}(\widetilde{D}^{*}\widetilde{D})^{-2}\big) \nonumber\\
  &=\int_{M}\int_{\|\xi\|=1} \operatorname{Tr}\left[\sigma_{-4}
\Big( \widetilde{c}(u) \widetilde{c}(v) \widetilde{c}(w)\widetilde{D}(\widetilde{D}^{*}\widetilde{D})^{-2}\Big)
\right](x, \xi)\sigma(\xi)\text{d}x.
 \end{align}

$\mathbf{Step~~I}$: primary symbol representation.

From (3.8) in \cite{WW6}, we get
\begin{align}\label{ABD}
& \sigma_{-4}
\Big( \widetilde{c}(u) \widetilde{c}(v) \widetilde{c}(w)\widetilde{D}(\widetilde{D}^{*}\widetilde{D})^{-2}\Big) \nonumber\\
=&\widetilde{c}(u) \widetilde{c}(v) \widetilde{c}(w)\bigg\{\sum_{|\alpha|=0}^{\infty} \frac{(-i)^{|\alpha|}}{\alpha!} \partial_{\xi}^{\alpha}(\sigma (\widetilde{D})) \partial_{x}^{\alpha}\big((\widetilde{D}^{*}\widetilde{D})^{-2}\big)\bigg\}_{-4} \nonumber\\
=&\widetilde{c}(u) \widetilde{c}(v) \widetilde{c}(w)\Big(
 \sigma_{0}(a_0\text{d}+b_0\delta+\sqrt{-1}c(X)) \sigma_{-4}((\widetilde{D}^{*}\widetilde{D})^{-2})+
  \sigma_{1}(a_0\text{d}+b_0\delta+\sqrt{-1}c(X)) \sigma_{-5}((\widetilde{D}^{*}\widetilde{D})^{-2})\nonumber\\
&-\sqrt{-1}\sum_{j=1}^{4}  \partial_{\xi_{j}} (\sigma_{1}(\widetilde{D})) \partial_{x_{j}}(\sigma_{-4}((\widetilde{D}^{*}\widetilde{D})^{-2}))\Big).
\end{align}

From $\widetilde{D}=a_0\text{d}+b_0\delta+\sqrt{-1}c(X)$, $\widetilde{D}^{*}=b_0\text{d}+a_0\delta+\sqrt{-1}c(X)$ and
$\widetilde{\Delta}= (b_0\text{d}+a_0\delta)(a_0\text{d}+b_0\delta)   =b_{0}^{2}d\delta+a_{0}^{2}\delta d $,
we obtain
\begin{align}
\widetilde{D}^{*}\widetilde{D}=\widetilde{\Delta}+(b_{0} d +a_{0} \delta   )\sqrt{-1}c(X)+\sqrt{-1}c(X)(a_0\text{d}+b_0\delta)+(\sqrt{-1}c(X))^{2}.
\end{align}
 Firstly, we compute the symbol expansion of $\widetilde{\Delta}=b_{0}^{2}d\delta+a_{0}^{2}\delta d$. The Lemma follows by direct computation.
\begin{lem}\label{le:32}
Let $\widetilde{\Delta}=b_{0}^{2}d\delta+a_{0}^{2}\delta d$ on $C^{\infty}(\Lambda^{k})$, then
\begin{align}
\sigma_{-2}(\widetilde{\Delta}^{-1})&=\frac{b_{0}^{2}|\xi|^{2}+(a_{0}^{2}-b_{0}^{2})\varepsilon(\xi)\iota(\xi)}{a_{0}^{2}b_{0}^{2}|\xi|^{4}};\\
\sigma_{-4}(\widetilde{\Delta}^{-2})&=\frac{b_{0}^{4}|\xi|^{2}+(a_{0}^{4}-b_{0}^{4})\varepsilon(\xi)\iota(\xi)}{a_{0}^{4}b_{0}^{4}|\xi|^{6}}.
\end{align}
\end{lem}
\begin{proof}
By using (7) in \cite{Ka}, we have
$\sigma_{2}(d\delta)=\varepsilon(\xi)\iota(\xi), \ \sigma_{2}(d\delta+\delta d)=|\xi|^{2}.$
Combining these results, we obtain
\begin{equation}
 \sigma_{2}(b_{0}^{2}d\delta+a_{0}^{2}\delta d)=(b_{0}^{2}-a_{0}^{2})\sigma_{2}(d\delta)+a_{0}^{2}\sigma_{2}(d\delta+\delta d)
=(b_{0}^{2}-a_{0}^{2})\varepsilon(\xi)\iota(\xi) +a_{0}^{2}|\xi|^{2}.
\end{equation}
Then an application of $|\xi|^{2}=\varepsilon(\xi)\iota(\xi)+\iota(\xi)\varepsilon(\xi)$, we obtain
\begin{equation}
[(b_{0}^{2}-a_{0}^{2})\varepsilon(\xi)\iota(\xi) +a_{0}^{2}|\xi|^{2}]
[b_{0}^{2}|\xi|^{2}+(a_{0}^{2}-b_{0}^{2})\varepsilon(\xi)\iota(\xi)]=a_{0}^{2}b_{0}^{2}|\xi|^{4}.
\end{equation}
Therefore
\begin{equation}
\sigma_{-2}(\widetilde{\Delta}^{-1})=\frac{b_{0}^{2}|\xi|^{2}+(a_{0}^{2}-b_{0}^{2})\varepsilon(\xi)\iota(\xi)}{a_{0}^{2}b_{0}^{2}|\xi|^{4}}.
\end{equation}
On the other hand, it is straightforward to see
\begin{equation}
\sigma_{-4}(\widetilde{\Delta}^{-2})=\Big(\sigma_{-2}(\widetilde{\Delta}^{-1})\Big)^{2}
=\frac{b_{0}^{4}|\xi|^{2}+(a_{0}^{4}-b_{0}^{4})\varepsilon(\xi)\iota(\xi)}{a_{0}^{4}b_{0}^{4}|\xi|^{6}}.
\end{equation}
\end{proof}
Write
\begin{equation}
\sigma(\widetilde{D})=p_1+p_0;
~\sigma(\widetilde{D}^{-1})=\sum^{\infty}_{j=1}q_{-j}.
\end{equation}
By the composition formula of psudodifferential operators,  we have
\begin{eqnarray*}
1=\sigma(\widetilde{D}\circ \widetilde{D}^{-1})&=&\sum_{\alpha}\frac{1}{\alpha!}\partial^{\alpha}_{\xi}[\sigma(\widetilde{D})]D^{\alpha}_{x}[\sigma(\widetilde{D}^{-1})]\\
&=&(p_1+p_0)(q_{-1}+q_{-2}+q_{-3}+\cdots)\\
& &~~~+\sum_j(\partial_{\xi_j}p_1+\partial_{\xi_j}p_0)(
D_{x_j}q_{-1}+D_{x_j}q_{-2}+D_{x_j}q_{-3}+\cdots)\\
&=&p_1q_{-1}+(p_1q_{-2}+p_0q_{-1}+\sum_j\partial_{\xi_j}p_1D_{x_j}q_{-1})+\cdots.
\end{eqnarray*}

\begin{lem}\label{le:32}
The following equations hold
\begin{align}
\sigma_{-3}(\widetilde{\Delta}^{-1})&= -\sigma_{-2}(\widetilde{\Delta}^{-1})\Big( \sigma_{1}(\widetilde{\Delta})\sigma_{-2}(\widetilde{\Delta}^{-1})
- \sqrt{-1}\sum_{j=1}^{4}\partial_{\xi_j} (\sigma_{2}(\widetilde{\Delta}) )\partial_{x_j}(\sigma_{-2}(\widetilde{\Delta}^{-1}) )  \Big);\\
\sigma_{-5}(\widetilde{\Delta}^{-2})&=
\sigma_{-3}(\widetilde{\Delta}^{-1})\sigma_{-2}(\widetilde{\Delta}^{-1})
+\sigma_{-2}(\widetilde{\Delta}^{-1})\sigma_{-3}(\widetilde{\Delta}^{-1})
- \sqrt{-1}\sum_{j=1}^{4}\partial_{\xi_j} (\sigma_{-2}(\widetilde{\Delta}^{-1}) )\partial_{x_j}(\sigma_{-2}(\widetilde{\Delta}^{-1}) )  .
\end{align}
\end{lem}

\begin{lem}\label{le:32}
The following equations hold
\begin{align}
\sigma_{-2}((\widetilde{D}^{*}\widetilde{D})^{-1})&= \sigma_{-2}(\widetilde{\Delta}^{-1});~~~~
\sigma_{-4}((\widetilde{D}^{*}\widetilde{D})^{-2})= \sigma_{-4}(\widetilde{\Delta}^{-2});\\
\sigma_{-5}((\widetilde{D}^{*}\widetilde{D})^{-2})&=
-\sigma_{-2}(\widetilde{\Delta}^{-1})\sigma_{1}((\widetilde{D}^{*}\widetilde{D})^{-1})\sigma_{-4}(\widetilde{\Delta}^{-2})
-\sigma_{-4}(\widetilde{\Delta}^{-2})\sigma_{1}((\widetilde{D}^{*}\widetilde{D})^{-1})\sigma_{-2}(\widetilde{\Delta}^{-1}) \nonumber\\
&- \sqrt{-1}\sum_{j=1}^{4}\partial_{\xi_j} (\sigma_{-2}(\widetilde{\Delta}^{-1}) )\partial_{x_j}(\sigma_{-2}(\widetilde{\Delta}^{-1}) )  .
\end{align}
\end{lem}

Let $g^{ij}=g(\text {d}x_{i},\text{d}x_{j})$, $\xi=\sum_{j}\xi_{j}dx_{j}$ and $\nabla^L_{\partial_{i}}\partial_{j}=\sum_{k}\Gamma_{ij}^{k}\partial_{k}$.
 In what follows, using the Einstein sum convention for repeated index summation:
\begin{align}\label{p2}
\sigma_{i}=-\frac{1}{4}\sum_{s,t}\omega_{s,t}
(e_i)c(e_s)c(e_t),~~ \xi^{j}=g^{ij}\xi_{i},~~\Gamma^{k}=g^{ij}\Gamma_{ij}^{k},~~\sigma^{j}=g^{ij}\sigma_{i}.
\end{align}
 Taking normal coordinates about $x_0$, then
$\sigma^i(x_0)=0,~
\partial^j[c(\partial_j)](x_0)=0 $,  $~\Gamma^k(x_0)=0$, $g^{ij}(x_0)=\delta_i^j$, $\partial^x_\mu g^{\alpha\beta}(x_0)=0.$
Then we get
$\partial_{x_j}\big(\sigma_{-2}(\widetilde{\Delta}^{-1})\big)(x_0)=0$, $\sigma_{-3}(\widetilde{\Delta}^{-1})(x_0)=0$.

From Lemma 3.3, Lemma 3.4 and Lemma 3.5, we obtain
\begin{align}\label{ABD}
& \sigma_{-4}
\Big( \widetilde{c}(u) \widetilde{c}(v) \widetilde{c}(w)\widetilde{D}(\widetilde{D}^{*}\widetilde{D})^{-2}\Big)(x_0) \nonumber\\
 =&\widetilde{c}(u) \widetilde{c}(v) \widetilde{c}(w)\Big(
 \sigma_{0}(a_0\text{d}+b_0\delta+\sqrt{-1}c(X)) \sigma_{-4}((\widetilde{D}^{*}\widetilde{D})^{-2})
 +  \sigma_{1}(a_0\text{d}+b_0\delta+\sqrt{-1}c(X)) \sigma_{-5}((\widetilde{D}^{*}\widetilde{D})^{-2}) \nonumber\\
&-\sqrt{-1}\sum_{j=1}^{4}  \partial_{\xi_{j}} (\sigma_{1}(\widetilde{D})) \partial_{x_{j}}(\sigma_{-4}((\widetilde{D}^{*}\widetilde{D})^{-2}))\Big)(x_0)\nonumber\\
 =&\widetilde{c}(u) \widetilde{c}(v) \widetilde{c}(w)\sqrt{-1}c(X)\sigma_{-4}(\Delta^{-2})(x_0)
 +\sqrt{-1}\widetilde{c}(u) \widetilde{c}(v) \widetilde{c}(w)\widetilde{c}( \xi)\sigma_{-5}((\widetilde{D}^{*}\widetilde{D})^{-2})(x_0)\nonumber\\
 =&H_{1}(x_0)+H_{2}(x_0),
\end{align}
where
\begin{align}\label{ABD}
H_{1}(x_0)=&\widetilde{c}(u) \widetilde{c}(v) \widetilde{c}(w)\sqrt{-1}c(X)\sigma_{-4}(\Delta^{-2})(x_0)\nonumber\\
=&\widetilde{c}(u) \widetilde{c}(v) \widetilde{c}(w)\sqrt{-1}c(X)
\frac{b_{0}^{4}|\xi|^{2}+(a_{0}^{4}-b_{0}^{4})\varepsilon(\xi)\iota(\xi)}{a_{0}^{4}b_{0}^{4}|\xi|^{6}};
\end{align}
and
\begin{align}\label{ABD}
H_{2}(x_0)=&\sqrt{-1}\widetilde{c}(u) \widetilde{c}(v) \widetilde{c}(w)\widetilde{c}( \xi)\sigma_{-5}((\widetilde{D}^{*}\widetilde{D})^{-2})(x_0)\nonumber\\
=&\sqrt{-1}\widetilde{c}(u) \widetilde{c}(v) \widetilde{c}(w)\widetilde{c}( \xi)
\Big(-\sigma_{-2}(\widetilde{\Delta}^{-1})\sigma_{1}((\widetilde{D}^{*}\widetilde{D})^{-1})\sigma_{-4}(\widetilde{\Delta}^{-2})\nonumber\\
&-\sigma_{-4}(\widetilde{\Delta}^{-2})\sigma_{1}((\widetilde{D}^{*}\widetilde{D})^{-1})\sigma_{-2}(\widetilde{\Delta}^{-1})\Big)(x_0)\nonumber\\
=& -\sqrt{-1}\widetilde{c}(u) \widetilde{c}(v) \widetilde{c}(w)\widetilde{c}( \xi)
\sigma_{-2}(\widetilde{\Delta}^{-1})\sigma_{1}((\widetilde{D}^{*}\widetilde{D})^{-1})\sigma_{-4}(\widetilde{\Delta}^{-2})(x_0)\nonumber\\
&-\sqrt{-1}\widetilde{c}(u) \widetilde{c}(v) \widetilde{c}(w)
\widetilde{c}(\xi)\sigma_{-4}(\widetilde{\Delta}^{-2})\sigma_{1}((\widetilde{D}^{*}\widetilde{D})^{-1})\sigma_{-2}(\widetilde{\Delta}^{-1}) (x_0),
\end{align}
where
\begin{align}
\sigma_{1}((\widetilde{D}^{*}\widetilde{D})^{-1})(x_0)
=&\sigma_{1}(b_0\text{d}+a_0\delta)\sqrt{-1}c(X)(x_0)+\sqrt{-1} c(X)\sigma_{1}(a_0\text{d}+b_0\delta)(x_0)\nonumber\\
=&\sqrt{-1}\bar{c}(\xi)\sqrt{-1}c(X)+\sqrt{-1}c(X)\sqrt{-1}\widetilde{c}(\xi),
\end{align}
Where $\bar{c}(\xi)=b_0\epsilon(\xi)-a_0\iota(\xi), \widetilde{c}(\xi)= a_0\epsilon(\xi)-b_0\iota(\xi).$

$\mathbf{Step~~II}$:  noncommutative residue representation.

Let $\widetilde{c}(e_j)=a_{0}\epsilon (e_j^* )-b_{0}\iota
(e_j)$,
  $e_j^*=g^{TM}(e_j,\cdot)$, and $\epsilon (e_j^*)$,~$\iota (e_j)$ be the exterior and interior multiplications respectively,
which satisfies
\begin{lem}
The following equations hold
\begin{align}\label{ee1}
&\widetilde{c}(X)c(Y)+c(Y)\widetilde{c}(X)=-(a_0+b_0)g^{TM}(X,Y);\nonumber\\
&\widetilde{c}(X)\widetilde{c}(Y)+\widetilde{c}(Y)\widetilde{c}(X)=-2a_0b_0g^{TM}(X,Y);\nonumber\\
&\widetilde{c}(X)\widehat{c}(Y)+\widehat{c}(Y)\widetilde{c}(X)=(a_0-b_0)g^{TM}(X,Y).
\end{align}
\end{lem}

 \begin{lem}
The following identities hold:
 \begin{align}
 &{\rm{Tr}}\big(\widetilde{c}(u) \widetilde{c}(v) \widetilde{c}(w)c(X)\big)
 = \frac{a_{0} b_{0}(a_{0}+b_{0})}{2}[g(u,X)g(v,w)-g(v,X)g(u,w)+g(w,X)g(u,v)]{\rm{Tr}}(\rm{Id}) ;\\
 &{\rm{Tr}}\big(\widetilde{c}(u) \widetilde{c}(v) \widetilde{c}(w)c(X)  \epsilon (\xi)\iota(\xi) \big)
=  \frac{a_{0} b_{0}|\xi|^{2}(a_{0}+b_{0})}{4}[g(u,X)g(v,w)-g(v,X)g(u,w)+g(w,X)g(u,v)]{\rm{Tr}}(\rm{Id})\nonumber\\
 &~~~~~~~~~~~~~~~~~~~~~~~~~~~~~~~~~~~~~+\frac{ (a_{0}^{2}b_{0}-a_{0}b_{0}^{2})}{4}\xi(X)[\xi(u)g(v,w)-\xi(v)g(u,w)+\xi(w)g(u,v)]{\rm{Tr}}(\rm{Id}).
\end{align}
\end{lem}
\begin{proof}
By Lemma (3.6) and the relation of the Clifford action and $ {\rm{Tr}}(AB)= {\rm{Tr}}(BA) $, we have the equality:
 \begin{align}
 {\rm{Tr}}\big(\widetilde{c}(u) \widetilde{c}(v) \widetilde{c}(w)c(X)\big)
 =& {\rm{Tr}}\big(\widetilde{c}(u) \widetilde{c}(v) (-c(X)\widetilde{c}(w)-(a_{0}+b_{0})g(X,w))\big)\nonumber\\
  =& -{\rm{Tr}}\big(\widetilde{c}(u) \widetilde{c}(v) c(X)\widetilde{c}(w)\big)
 -(a_{0}+b_{0})g(X,w) {\rm{Tr}}\big(\widetilde{c}(u) \widetilde{c}(v) \big)\nonumber\\
  =& \cdots \nonumber\\
    =& - {\rm{Tr}}\big(\widetilde{c}(u) \widetilde{c}(v) \widetilde{c}(w)c(X)\big) -(a_{0}+b_{0})g(X,u) {\rm{Tr}}\big(\widetilde{c}(v) \widetilde{c}(w) \big)\nonumber\\
     &+(a_{0}+b_{0})g(X,v) {\rm{Tr}}\big(\widetilde{c}(u) \widetilde{c}(w) \big)-(a_{0}+b_{0})g(X,w) {\rm{Tr}}\big(\widetilde{c}(u) \widetilde{c}(v) \big).
\end{align}
By transferring the first term on the right side of (3.27) to the left, the first equation of this Lemma can be obtained.

Similarly, it can be directly calculated
\begin{align}
  &{\rm{Tr}}\big(\widetilde{c}(u) \widetilde{c}(v) \widetilde{c}(w)c(X)  \epsilon (\xi)\iota(\xi) \big)\nonumber\\
 =&{\rm{Tr}}\big(-\widetilde{c}(u) \widetilde{c}(v) (c(X)\widetilde{c}(w)-(a_{0}+b_{0})g(X,w))  \epsilon (\xi)\iota(\xi) \big)\nonumber\\
  =&-{\rm{Tr}}\big(\widetilde{c}(u) \widetilde{c}(v)  c(X)\widetilde{c}(w)   \epsilon (\xi)\iota(\xi) \big)
   -(a_{0}+b_{0})g(X,w){\rm{Tr}}\big(\widetilde{c}(u) \widetilde{c}(v)   \epsilon (\xi)\iota(\xi) \big)\nonumber\\
     =& \cdots \nonumber\\
 =&-{\rm{Tr}}\big(\widetilde{c}(u) \widetilde{c}(v) \widetilde{c}(w)c(X)  \epsilon (\xi)\iota(\xi) \big)
 -(a_{0}+b_{0})g(X,w){\rm{Tr}}\big(\widetilde{c}(u) \widetilde{c}(v)   \epsilon (\xi)\iota(\xi) \big)\nonumber\\
 &-(a_{0}+b_{0})g(X,u){\rm{Tr}}\big(\widetilde{c}(v) \widetilde{c}(w)   \epsilon (\xi)\iota(\xi) \big) +(a_{0}+b_{0})g(X,v){\rm{Tr}}\big(\widetilde{c}(u) \widetilde{c}(w)   \epsilon (\xi)\iota(\xi) \big)\nonumber\\
& -\xi(X){\rm{Tr}}\big(\widetilde{c}(u) \widetilde{c}(v)\widetilde{c}(w)  \iota(\xi) \big)
-\xi(X){\rm{Tr}}\big(\widetilde{c}(u) \widetilde{c}(v)\widetilde{c}(w)  \epsilon (\xi) \big).
\end{align}
Similarly, we obtain
\begin{align}
 &{\rm{Tr}}\big(\widetilde{c}(u) \widetilde{c}(v)   \epsilon (\xi)\iota(\xi) \big)=-\frac{a_{0} b_{0}}{2}|\xi|^{2}g(u,v){\rm{Tr}}(\rm{Id});
  {\rm{Tr}}\big(\widetilde{c}(v) \widetilde{c}(w)   \epsilon (\xi)\iota(\xi) \big)=-\frac{a_{0} b_{0}}{2}|\xi|^{2}g(v,w){\rm{Tr}}(\rm{Id});\nonumber\\
& {\rm{Tr}}\big(\widetilde{c}(u) \widetilde{c}(w)   \epsilon (\xi)\iota(\xi) \big)=-\frac{a_{0} b_{0}}{2}|\xi|^{2}g(u,w){\rm{Tr}}(\rm{Id});\nonumber\\
& {\rm{Tr}}\big(\widetilde{c}(u) \widetilde{c}(v)\widetilde{c}(w)  \iota(\xi) \big)
 =\frac{ -a_{0}^{2}b_{0}}{4}\big[\xi(u)g(v,w)-\xi(v)g(u,w)+\xi(w)g(u,v)\big]{\rm{Tr}}(\rm{Id});\nonumber\\
 &{\rm{Tr}}\big(\widetilde{c}(u) \widetilde{c}(v)\widetilde{c}(w)  \epsilon (\xi) \big)
 =\frac{ a_{0}b_{0}^{2}}{4}\big[\xi(u)g(v,w)-\xi(v)g(u,w)+\xi(w)g(u,v)\big]{\rm{Tr}}(\rm{Id}).
\end{align}
By  (3.28) and (3.29), the second  equation of this Lemma can be obtained.
\end{proof}

 \begin{lem}
The following identities hold:
 \begin{align}
 (1)~~&{\rm{Tr}}\Big(\widetilde{c}(u) \widetilde{c}(v) \widetilde{c}(w)\widetilde{c}(\xi)\big( \bar{c}(\xi)c(X)+c(X)\widetilde{c}(\xi) \big)\Big)\nonumber\\
 =& -a_{0}^{2}b_{0}^{2}(a_{0}+b_{0})\xi(X)[\xi(u)g(v,w)-\xi(v)g(u,w)+\xi(w)g(u,v)]{\rm{Tr}}(\rm{Id});\\
 (2)~~ &{\rm{Tr}}\Big(\widetilde{c}(u) \widetilde{c}(v) \widetilde{c}(w)\widetilde{c}(\xi)\big( \bar{c}(\xi)c(X) \epsilon (\xi)\iota(\xi)+c(X)\widetilde{c}(\xi) \epsilon (\xi)\iota(\xi) \big)\Big)\nonumber\\
 =&\frac{a_{0}^{4}b_{0} -5a_{0}^{2}b_{0}^{3} }{4}
 |\xi|^{2}\xi(X)\big[\xi(u)g(v,w)-\xi(v)g(u,w)+\xi(w)g(u,v)\big]{\rm{Tr}}(\rm{Id})\nonumber\\
   &-\frac{a_{0}^{2}b_{0}(a_{0}+ b_{0})(a_{0}- b_{0})}{4}|\xi|^{4}\big[g(X,w)g(u,v)-g(X,v)g(u,w)+g(X,u)g(v,w)\big]{\rm{Tr}}(\rm{Id});\\
   (3)~~  &{\rm{Tr}}\Big(\widetilde{c}(u) \widetilde{c}(v) \widetilde{c}(w)\widetilde{c}(\xi)\big( \epsilon (\xi)\iota(\xi)\bar{c}(\xi)c(X)+ \epsilon (\xi)\iota(\xi)c(X)\widetilde{c}(\xi) \big)\Big)\nonumber\\
 =&-a_{0}^{2}b_{0}^{3}|\xi|^{2}\xi(X)[\xi(u)g(v,w)-\xi(v)g(u,w)+\xi(w)g(u,v)]{\rm{Tr}}(\rm{Id});\\
   (4)~~  &{\rm{Tr}}\Big(\widetilde{c}(u) \widetilde{c}(v) \widetilde{c}(w)\widetilde{c}(\xi)\big( \epsilon (\xi)\iota(\xi)\bar{c}(\xi)c(X)\epsilon (\xi)\iota(\xi)+ \epsilon (\xi)\iota(\xi)c(X)\widetilde{c}(\xi)\epsilon (\xi)\iota(\xi) \big)\Big)\nonumber\\
 =&-a_{0}^{2}b_{0}^{3}|\xi|^{4}\xi(X)[\xi(u)g(v,w)-\xi(v)g(u,w)+\xi(w)g(u,v)]{\rm{Tr}}(\rm{Id}) .
\end{align}
\end{lem}
\begin{proof}
For $\bar{c}(\xi)=b_0\epsilon(\xi)-a_0\iota(\xi), \widetilde{c}(\xi)= a_0\epsilon(\xi)-b_0\iota(\xi)$, we get
 \begin{align}
\bar{c}(\xi)c(X)+c(X)\widetilde{c}(\xi)=(a_0-b_0)\big( \epsilon (X)\epsilon (\xi)-\iota(X)\iota(\xi)+\epsilon (X)\iota(\xi)-\iota(X) \epsilon (\xi) \big)
-(a_0+b_0) \xi(X).
\end{align}

(1).  By Lemma (3.6) and the relation of the Clifford action and $ {\rm{Tr}}(AB)= {\rm{Tr}}(BA) $, then
   \begin{align}
&{\rm{Tr}}\Big(\widetilde{c}(u) \widetilde{c}(v) \widetilde{c}(w)\widetilde{c}(\xi)\big( \bar{c}(\xi)c(X)+c(X)\widetilde{c}(\xi) \big)\Big)\nonumber\\
 =&  {\rm{Tr}}\Big(\widetilde{c}(u) \widetilde{c}(v) \widetilde{c}(w)\widetilde{c}(\xi) \big[(a_0-b_0)\big( \epsilon (X)\epsilon (\xi)-\iota(X)\iota(\xi)+\epsilon (X)\iota(\xi)-\iota(X) \epsilon (\xi) \big)
-(a_0+b_0) \xi(X)\big]\Big)\nonumber\\
 =&  (a_0-b_0){\rm{Tr}}\Big(\widetilde{c}(u) \widetilde{c}(v) \widetilde{c}(w)\widetilde{c}(\xi)c(X) \epsilon (\xi)\Big)
 + (a_0-b_0){\rm{Tr}}\Big(\widetilde{c}(u) \widetilde{c}(v) \widetilde{c}(w)\widetilde{c}(\xi) c(X) \iota(\xi)\Big)\nonumber\\
 &-(a_0+b_0) \xi(X){\rm{Tr}}\Big(\widetilde{c}(u) \widetilde{c}(v) \widetilde{c}(w)\widetilde{c}(\xi)\Big)\nonumber\\
 =& -(a_0+b_0) \xi(X)a_{0}^{2}b_{0}^{2}\big[\xi(u)g(v,w)-\xi(v)g(u,w)+\xi(w)g(u,v)\big]{\rm{Tr}}(\rm{Id}).
\end{align}

(2). We note that
  \begin{align}
\big(\bar{c}(\xi)c(X)+c(X)\widetilde{c}(\xi)\big) \epsilon (\xi)\iota(\xi)=(a_0-b_0) |\xi|^{2}c(X)\iota(\xi)-(a_0+b_0)  \xi(X)\epsilon (\xi)\iota(\xi).
\end{align}
Then
   \begin{align}
&{\rm{Tr}}\Big(\widetilde{c}(u) \widetilde{c}(v) \widetilde{c}(w)\widetilde{c}(\xi)\big( \bar{c}(\xi)c(X) \epsilon (\xi)\iota(\xi)+c(X)\widetilde{c}(\xi) \epsilon (\xi)\iota(\xi) \big)\Big)\nonumber\\
=& (a_0-b_0) |\xi|^{2}{\rm{Tr}}\Big(\widetilde{c}(u) \widetilde{c}(v) \widetilde{c}(w)\widetilde{c}(\xi)c(X)\iota(\xi)\Big)
-(a_0+b_0)  \xi(X) {\rm{Tr}}\Big(\widetilde{c}(u) \widetilde{c}(v) \widetilde{c}(w)\widetilde{c}(\xi)\epsilon (\xi)\iota(\xi)\Big)
\end{align}
By Lemma (3.6) and the relation of the Clifford action and $ {\rm{Tr}}(AB)= {\rm{Tr}}(BA) $,
\begin{align}
  &{\rm{Tr}}\big(\widetilde{c}(u) \widetilde{c}(v) \widetilde{c}(w)\widetilde{c}(\xi)c(X)\iota(\xi) \big)\nonumber\\
 =&{\rm{Tr}}\big(\widetilde{c}(u) \widetilde{c}(v)  \widetilde{c}(w)
    (-c(X)\widetilde{c}(\xi)-(a_{0}+b_{0})\xi(X)) \iota(\xi) \big)\nonumber\\
  =&-{\rm{Tr}}\big(\widetilde{c}(u) \widetilde{c}(v) \widetilde{c}(w)c(X)\widetilde{c}(\xi)\iota(\xi) \big)
  -(a_{0}+b_{0})\xi(X) {\rm{Tr}}\big(\widetilde{c}(u) \widetilde{c}(v)  \widetilde{c}(w) \iota(\xi) \big)\nonumber\\
=& \cdots \nonumber\\
   =&-{\rm{Tr}}\big(\widetilde{c}(u) \widetilde{c}(v) \widetilde{c}(w)\widetilde{c}(\xi)c(X)\iota(\xi) \big)
 +\xi(X){\rm{Tr}}\big(\widetilde{c}(u) \widetilde{c}(v)\widetilde{c}(w) \widetilde{c}(\xi) \big)\nonumber\\
 &+(a_{0}+b_{0})g(X,u){\rm{Tr}}\big(\widetilde{c}(v) \widetilde{c}(w)\widetilde{c}(\xi)\iota(\xi) \big)
 -(a_{0}+b_{0})g(X,v){\rm{Tr}}\big(\widetilde{c}(u) \widetilde{c}(w) \widetilde{c}(\xi)\iota(\xi) \big) \nonumber\\
 &+(a_{0}+b_{0})g(X,w){\rm{Tr}}\big(\widetilde{c}(u) \widetilde{c}(v)\widetilde{c}(\xi)\iota(\xi) \big)
  -(a_{0}+b_{0})\xi(X){\rm{Tr}}\big(\widetilde{c}(u)  \widetilde{c}(v)\widetilde{c}(w) \iota(\xi) \big).
 \end{align}

Similarly, we obtain
\begin{align}
 &{\rm{Tr}}\big(\widetilde{c}(u) \widetilde{c}(v)\widetilde{c}(w) \widetilde{c}(\xi) \big)
 =a_{0}^{2}b_{0}^{2}\big[\xi(u)g(v,w)-\xi(v)g(u,w)+\xi(w)g(u,v)\big]{\rm{Tr}}(\rm{Id});\nonumber\\
& {\rm{Tr}}\big(\widetilde{c}(v) \widetilde{c}(w)\widetilde{c}(\xi)\iota(\xi) \big)
=-\frac{a_{0}^{2} b_{0}}{2}|\xi|^{2}g(v,w){\rm{Tr}}(\rm{Id}) ;\nonumber\\
& {\rm{Tr}}\big(\widetilde{c}(u) \widetilde{c}(w) \widetilde{c}(\xi)\iota(\xi) \big)
 =-\frac{a_{0}^{2} b_{0}}{2}|\xi|^{2}g(u,w){\rm{Tr}}(\rm{Id});\nonumber\\
 &{\rm{Tr}}\big(\widetilde{c}(u) \widetilde{c}(v)\widetilde{c}(\xi)\iota(\xi) \big)
 =-\frac{a_{0}^{2} b_{0}}{2}|\xi|^{2}g(u,v){\rm{Tr}}(\rm{Id});\nonumber\\
 & {\rm{Tr}}\big(\widetilde{c}(u)  \widetilde{c}(v)\widetilde{c}(w) \iota(\xi) \big)
 =\frac{ -a_{0}^{2}b_{0}}{2}\big[\xi(u)g(v,w)-\xi(v)g(u,w)+\xi(w)g(u,v)\big]{\rm{Tr}}(\rm{Id}).
\end{align}
By (3.38) and (3.39), we obtain
\begin{align}
  &{\rm{Tr}}\big(\widetilde{c}(u) \widetilde{c}(v) \widetilde{c}(w)\widetilde{c}(\xi)c(X)\iota(\xi) \big)\nonumber\\
   =&\frac{a_{0}^{2}b_{0}(a_{0}+3b_{0})}{4}\xi(X)\big[\xi(u)g(v,w)-\xi(v)g(u,w)+\xi(w)g(u,v)\big]{\rm{Tr}}(\rm{Id})\nonumber\\
  &-\frac{a_{0}^{2}b_{0}(a_{0}+ b_{0})}{4}|\xi|^{2}\big[g(X,w)g(u,v)-g(X,v)g(u,w)+g(X,u)g(v,w)\big]{\rm{Tr}}(\rm{Id}).
 \end{align}
 Similarly, we obtain
 \begin{align}
{\rm{Tr}}\Big(\widetilde{c}(u) \widetilde{c}(v) \widetilde{c}(w)\widetilde{c}(\xi)\epsilon (\xi)\iota(\xi)\Big)
   =\frac{a_{0}^{2}b_{0}^{2} }{2}|\xi|^{2} \big[\xi(u)g(v,w)-\xi(v)g(u,w)+\xi(w)g(u,v)\big]{\rm{Tr}}(\rm{Id}) .
 \end{align}
 Combining equations (3.40) and (3.41) yields the second equation of this Lemma.

(3). It is obtained through direct calculation
  \begin{align}
\epsilon (\xi)\iota(\xi)\big(\bar{c}(\xi)c(X)+c(X)\widetilde{c}(\xi)\big)=(a_0-b_0)|\xi|^{2}\big(\xi(X)+  c(X)\epsilon (\xi) \big)-(a_0+b_0)  \xi(X)\epsilon (\xi)\iota(\xi).
\end{align}
Then
 \begin{align}
  &{\rm{Tr}}\Big(\widetilde{c}(u) \widetilde{c}(v) \widetilde{c}(w)\widetilde{c}(\xi)\big( \epsilon (\xi)\iota(\xi)\bar{c}(\xi)c(X)+ \epsilon (\xi)\iota(\xi)c(X)\widetilde{c}(\xi) \big)\Big)\nonumber\\
  =& {\rm{Tr}}\Big(\widetilde{c}(u) \widetilde{c}(v) \widetilde{c}(w)\widetilde{c}(\xi)
  \big( (a_0-b_0)|\xi|^{2}\big(\xi(X)+  c(X)\epsilon (\xi) \big)-(a_0+b_0)  \xi(X)\epsilon (\xi)\iota(\xi)
   \big)\Big)\nonumber\\
 =& (a_0-b_0)|\xi|^{2}\xi(X){\rm{Tr}}\Big(\widetilde{c}(u) \widetilde{c}(v) \widetilde{c}(w)\widetilde{c}(\xi)
  \Big)
  +(a_0-b_0)|\xi|^{2}  {\rm{Tr}}\Big(\widetilde{c}(u) \widetilde{c}(v) \widetilde{c}(w)\widetilde{c}(\xi)
 c(X)\epsilon (\xi) \Big)\nonumber\\
   &  -(a_0+b_0) \xi(X)  {\rm{Tr}}\Big(\widetilde{c}(u) \widetilde{c}(v) \widetilde{c}(w)\widetilde{c}(\xi)
\epsilon (\xi)\iota(\xi)
 \Big)\nonumber\\
=&-a_{0}^{2}b_{0}^{3}|\xi|^{2}\xi(X)[\xi(u)g(v,w)-\xi(v)g(u,w)+\xi(w)g(u,v)]{\rm{Tr}}(\rm{Id}).
\end{align}

(4). We note that
  \begin{align}
\epsilon (\xi)\iota(\xi)\big(\bar{c}(\xi)c(X)+c(X)\widetilde{c}(\xi)\big) \epsilon (\xi)\iota(\xi)= -2b_0  |\xi|^{2}\xi(X)\epsilon (\xi)\iota(\xi).
\end{align}
 Then
 \begin{align}
 &{\rm{Tr}}\Big(\widetilde{c}(u) \widetilde{c}(v) \widetilde{c}(w)\widetilde{c}(\xi)\big( \epsilon (\xi)\iota(\xi)\bar{c}(\xi)c(X)\epsilon (\xi)\iota(\xi)+ \epsilon (\xi)\iota(\xi)c(X)\widetilde{c}(\xi)\epsilon (\xi)\iota(\xi) \big)\Big)\nonumber\\
 = &{\rm{Tr}}\Big(\widetilde{c}(u) \widetilde{c}(v) \widetilde{c}(w)\widetilde{c}(\xi)\big( -2b_0  |\xi|^{2}\xi(X)\epsilon (\xi)\iota(\xi)\big)\Big)\nonumber\\
  = &-2b_0  |\xi|^{2}\xi(X){\rm{Tr}}\Big(\widetilde{c}(u) \widetilde{c}(v) \widetilde{c}(w)\widetilde{c}(\xi) \epsilon (\xi)\iota(\xi) \Big)\nonumber\\
 =&-a_{0}^{2}b_{0}^{3}|\xi|^{4}\xi(X)\big[\xi(u)g(v,w)-\xi(v)g(u,w)+\xi(w)g(u,v)\big]{\rm{Tr}}(\rm{Id}) .
\end{align}
 \end{proof}
By integrating formula we get
\begin{align}
& \int_{\|\xi\|=1} \operatorname{Tr} [H_{1}  (x_{0} )+H_{2}  (x_{0} ) ](x, \xi)  \sigma(\xi)  \nonumber\\
=&\int_{\|\xi\|=1} \operatorname{Tr} [\widetilde{c}(u) \widetilde{c}(v) \widetilde{c}(w)\sqrt{-1}c(X)
\frac{b_{0}^{4}|\xi|^{2}+(a_{0}^{4}-b_{0}^{4})\varepsilon(\xi)\iota(\xi)}{a_{0}^{4}b_{0}^{4}|\xi|^{6}}  (x_{0} ) ](x, \xi)  \sigma(\xi)  \nonumber\\
&+\int_{\|\xi\|=1} \operatorname{Tr} \Big[ -\sqrt{-1}\widetilde{c}(u) \widetilde{c}(v) \widetilde{c}(w)\widetilde{c}( \xi)
\sigma_{-2}(\widetilde{\Delta}^{-1})\sigma_{1}((\widetilde{D}^{*}\widetilde{D})^{-1})\sigma_{-4}(\widetilde{\Delta}^{-2})  (x_{0} ) \Big](x, \xi)  \sigma(\xi)  \nonumber\\
&+\int_{\|\xi\|=1} \operatorname{Tr} \Big[  -\sqrt{-1}\widetilde{c}(u) \widetilde{c}(v) \widetilde{c}(w)
\widetilde{c}(\xi)\sigma_{-4}(\widetilde{\Delta}^{-2})\sigma_{1}((\widetilde{D}^{*}\widetilde{D})^{-1})\sigma_{-2}(\widetilde{\Delta}^{-1}) (x_{0} ) \Big](x, \xi)  \sigma(\xi)  \nonumber\\
=&\frac{3\sqrt{-1}(a_{0}^{4}-b_{0}^{4})(2b_{0}^{2}-3a_{0}^{2}-a_{0}b_{0})}{16a_{0}^{4}b_{0}^{3}}
\Big[g(u,X)g(v,w)-g(v,X)g(u,w)+g(w,X)g(u,v)\Big]{\rm{Tr}}(\rm{Id}).
 \end{align}
Then we obtain the spectral torsion associated with nonminimal de Rham-Hodge operators.
 \begin{thm}\label{mrthm}
 	Let $M$ be an four-dimensional oriented compact spin Riemannian manifold,  with the trilinear Clifford multiplication by functional of differential one-forms $\widetilde{c}(u),\widetilde{c}(v),\widetilde{c}(w),$ the spectral torsion $\mathscr{T}_{\widetilde{D}}$ for  nonminimal de Rham-Hodge operators  $\widetilde{D}=a_0\text{d}+b_0\delta+\sqrt{-1}c(X),  \widetilde{D}^{*}=b_0\text{d}+a_0\delta+\sqrt{-1}c(X)$  equals to
\begin{align}
 	&\mathscr{T}_{\widetilde{D}}\big(\widetilde{c}(u),\widetilde{c}(v),\widetilde{c}(w)\big)
= \mathrm{Wres}\big(\widetilde{c}(u) \widetilde{c}(v) \widetilde{c}(w) \widetilde{D}(\widetilde{D}^{*}\widetilde{D})^{-2}\big) \nonumber\\
=& \int_{M}\bigg\{
\frac{12\sqrt{-1}(a_{0}^{4}-b_{0}^{4})(2b_{0}^{2}-3a_{0}^{2}-a_{0}b_{0})}{16a_{0}^{4}b_{0}^{3}}
\Big[g(u,X)g(v,w)-g(v,X)g(u,w)+g(w,X)g(u,v)\Big] {\rm Vol}(S^{3})
\bigg\}{\rm d Vol}_M.
\end{align}
 \end{thm}

 \section{The spectral one form for  nonminimal de-Rham Hodge operators}

The purpose of this section is to specify the de-Rham Hodge type operator and compute the spectral form for  nonminimal de-Rham Hodge operators.
To avoid technique terminology we only state our results for
compact oriented Riemannian manifold of even dimension $n=2m$ with a Riemannian metric $g^{M}$ by using the trace of  nonminimal de-Rham Hodge operators and the noncommutative residue density.  By  Definition 2.1 in \cite{DL2} and Definition 2.2 in \cite{WW3}, we have
\begin{defn}
For the  nonminimal  de-Rham Hodge type operator $\widetilde{D}=a_0\text{d}+b_0\delta+\sqrt{-1}c(X)$ and
 the adjoint operator $\widetilde{D}^{*}=b_0\text{d}+a_0\delta+\sqrt{-1}c(X)$, the trilinear Clifford multiplication by functional of differential one-forms $\widehat{c}(u)$
\begin{align}
\mathscr{T}_{1}(\widetilde{c}(u)  )
=\mathrm{Wres}\big(\widetilde{c}(u)  \widetilde{D}(\widetilde{D}^{*}\widetilde{D})^{-m}\big)
=\int_{M}\int_{\|\xi\|=1} \operatorname{Tr}\left[\sigma_{-2 m}
\Big( \widetilde{c}(u)  \widetilde{D}(\widetilde{D}^{*}\widetilde{D})^{-m}\Big)
\right](x, \xi)\sigma(\xi)\text{d}x.
 \end{align}
is called thespectral one form $\mathscr{T}_{1}$.
\end{defn}

For $n=4$, from Lemma 3.3, Lemma 3.4 and Lemma 3.5, we obtain
\begin{align}\label{ABD}
& \sigma_{-4}
\Big( \widetilde{c}(u)  \widetilde{D}(\widetilde{D}^{*}\widetilde{D})^{-2}\Big)(x_0) \nonumber\\
 =&\widetilde{c}(u)  \Big(
 \sigma_{0}(a_0\text{d}+b_0\delta+\sqrt{-1}c(X)) \sigma_{-4}((\widetilde{D}^{*}\widetilde{D})^{-2})
 +  \sigma_{1}(a_0\text{d}+b_0\delta+\sqrt{-1}c(X)) \sigma_{-5}((\widetilde{D}^{*}\widetilde{D})^{-2}) \nonumber\\
&-\sqrt{-1}\sum_{j=1}^{4}  \partial_{\xi_{j}} (\sigma_{1}(\widetilde{D})) \partial_{x_{j}}(\sigma_{-4}((\widetilde{D}^{*}\widetilde{D})^{-2}))\Big)(x_0)\nonumber\\
 =&\widetilde{c}(u)  \sqrt{-1}c(X)\sigma_{-4}(\Delta^{-2})(x_0)
 +\sqrt{-1}\widetilde{c}(u) \widetilde{c}( \xi)\sigma_{-5}((\widetilde{D}^{*}\widetilde{D})^{-2})(x_0)\nonumber\\
 =&L_{1}(x_0)+L_{2}(x_0),
\end{align}
where
\begin{align}\label{ABD}
L_{1}(x_0)= \widetilde{c}(u) \sqrt{-1}c(X)\sigma_{-4}(\Delta^{-2})(x_0)
= \widetilde{c}(u)  \sqrt{-1}c(X)
\frac{b_{0}^{4}|\xi|^{2}+(a_{0}^{4}-b_{0}^{4})\varepsilon(\xi)\iota(\xi)}{a_{0}^{4}b_{0}^{4}|\xi|^{6}};
\end{align}
and
\begin{align}\label{ABD}
L_{2}(x_0)=&\sqrt{-1}\widetilde{c}(u)  \widetilde{c}( \xi)\sigma_{-5}((\widetilde{D}^{*}\widetilde{D})^{-2})(x_0)\nonumber\\
=&\sqrt{-1}\widetilde{c}(u)  \widetilde{c}( \xi)
\Big(-\sigma_{-2}(\widetilde{\Delta}^{-1})\sigma_{1}((\widetilde{D}^{*}\widetilde{D})^{-1})\sigma_{-4}(\widetilde{\Delta}^{-2})\nonumber\\
&-\sigma_{-4}(\widetilde{\Delta}^{-2})\sigma_{1}((\widetilde{D}^{*}\widetilde{D})^{-1})\sigma_{-2}(\widetilde{\Delta}^{-1})\Big)(x_0)\nonumber\\
=& -\sqrt{-1}\widetilde{c}(u)  \widetilde{c}( \xi)
\sigma_{-2}(\widetilde{\Delta}^{-1})\sigma_{1}((\widetilde{D}^{*}\widetilde{D})^{-1})\sigma_{-4}(\widetilde{\Delta}^{-2})(x_0)\nonumber\\
&-\sqrt{-1}\widetilde{c}(u)
\widetilde{c}(\xi)\sigma_{-4}(\widetilde{\Delta}^{-2})\sigma_{1}((\widetilde{D}^{*}\widetilde{D})^{-1})\sigma_{-2}(\widetilde{\Delta}^{-1}) (x_0),
\end{align}
where
\begin{align}
\sigma_{1}((\widetilde{D}^{*}\widetilde{D})^{-1})(x_0)
=&\sqrt{-1}(b_0\text{d}+a_0\delta)\sqrt{-1}c(X)+(\sqrt{-1})^{2}c(X)(a_0\text{d}+b_0\delta)(x_0)\nonumber\\
=&\sqrt{-1}\bar{c}(\xi)\sqrt{-1}c(X)+(\sqrt{-1})^{2}c(X)\widetilde{c}(\xi).
\end{align}

The following two lemmas play an important role in our calculation of
the spectral form  associated with the nonminimal  de-Rham Hodge type operator.

 \begin{lem}
The following identities hold:
 \begin{align}
 &{\rm{Tr}}\big(\widetilde{c}(u) c(X)\big)
 = \frac{ -(a_{0}+b_{0})}{2}g(u,X){\rm{Tr}}(\rm{Id}) ;\\
 &{\rm{Tr}}\big(\widetilde{c}(u)  c(X)  \epsilon (\xi)\iota(\xi) \big)
= \frac{1}{4}\big( -a_{0}\xi(X)\xi(u)+b_{0}\xi(X)\xi(u)-(a_{0}+b_{0})|\xi|^{2} g(u,X)\big){\rm{Tr}}(\rm{Id}).
\end{align}
\end{lem}
\begin{proof}
By Lemma (3.6), we have the equality:
\begin{align}
\widetilde{c}(u) c(X)+c(X)\widetilde{c}(u)=-(a_0+b_0)g^{TM}( u,X),
\end{align}
By the relation of the Clifford action and $ {\rm{Tr}}(AB)= {\rm{Tr}}(BA) $, we get
 \begin{align}
{\rm{Tr}}\big(\widetilde{c}(u) c(X)\big)
 = \frac{ -(a_{0}+b_{0})}{2}g(u,X){\rm{Tr}}(\rm{Id}).
\end{align}
Similarly, it can be directly calculated
\begin{align}
  &{\rm{Tr}}\big(\widetilde{c}(u)  c(X)  \epsilon (\xi)\iota(\xi) \big)\nonumber\\
 =&{\rm{Tr}}\big( (-c(X)\widetilde{c}(u)-(a_{0}+b_{0})g(u,X))  \epsilon (\xi)\iota(\xi) \big)\nonumber\\
  =&-{\rm{Tr}}\big(  c(X)\widetilde{c}(u)   \epsilon (\xi)\iota(\xi) \big)
   -(a_{0}+b_{0})g(u,X){\rm{Tr}}\big( \epsilon (\xi)\iota(\xi) \big)\nonumber\\
     =& \cdots \nonumber\\
 =&-{\rm{Tr}}\big(\widetilde{c}(u)c(X)  \epsilon (\xi)\iota(\xi) \big)
 -(a_{0}+b_{0})g(u,X){\rm{Tr}}\big( \epsilon (\xi)\iota(\xi) \big)
- \xi(X){\rm{Tr}}\big(\widetilde{c}(u)\iota(\xi) \big)- \xi(X){\rm{Tr}}\big(\widetilde{c}(u)\epsilon (\xi) \big)
\end{align}
By transferring the first term on the right side of (4.10) to the left, the second equation of this Lemma can be obtained.
\end{proof}

 \begin{lem}
The following identities hold:
 \begin{align}
 (1)~~&{\rm{Tr}}\Big(\widetilde{c}(u)  \widetilde{c}(\xi)\big( \bar{c}(\xi)c(X)+c(X)\widetilde{c}(\xi) \big)\Big)
 =  a_{0} b_{0} (a_{0}+b_{0})\xi(X) \xi(u) {\rm{Tr}}(\rm{Id});\\
 (2)~~ &{\rm{Tr}}\Big(\widetilde{c}(u)  \widetilde{c}(\xi)\big( \bar{c}(\xi)c(X) \epsilon (\xi)\iota(\xi)+c(X)\widetilde{c}(\xi) \epsilon (\xi)\iota(\xi) \big)\Big)\nonumber\\
 =&\big(\frac{a_{0}(3b_{0}^{2} - a_{0}^{2})}{2} |\xi|^{2}\xi(X)\xi(u)+\frac{a_{0}(a_{0}^{2} - b_{0}^{2})}{2} |\xi|^{2}g(u,v)\big){\rm{Tr}}(\rm{Id});\\
   (3)~~  &{\rm{Tr}}\Big(\widetilde{c}(u)  \widetilde{c}(\xi)\big( \epsilon (\xi)\iota(\xi)\bar{c}(\xi)c(X)+ \epsilon (\xi)\iota(\xi)c(X)\widetilde{c}(\xi) \big)\Big)\nonumber\\
 =&\big(\frac{a_{0}^{3}b_{0}-a_{0}b_{0}^{3}+2a_{0}b_{0}^{2} }{2}|\xi|^{2}\xi(X) \xi(u)
 - \frac{b_{0}(a_{0}^{2} - b_{0}^{2})}{4} |\xi|^{2}g(u,v)\big){\rm{Tr}}(\rm{Id});\\
   (4)~~  &{\rm{Tr}}\Big(\widetilde{c}(u) \widetilde{c}(\xi)\big( \epsilon (\xi)\iota(\xi)\bar{c}(\xi)c(X)\epsilon (\xi)\iota(\xi)+ \epsilon (\xi)\iota(\xi)c(X)\widetilde{c}(\xi)\epsilon (\xi)\iota(\xi) \big)\Big) \nonumber\\
 =&a_{0} b_{0}^{2 }(a_{0}+b_{0})|\xi|^{4}\xi(X) \xi(u) {\rm{Tr}}(\rm{Id}).
\end{align}
\end{lem}
\begin{proof}
For $\bar{c}(\xi)=b_0\epsilon(\xi)-a_0\iota(\xi), \widetilde{c}(\xi)= a_0\epsilon(\xi)-b_0\iota(\xi),$  we get
\begin{align}
\bar{c}(\xi)c(X)+c(X)\widetilde{c}(\xi)=(a_0-b_0)\big( \epsilon (X)\epsilon (\xi)-\iota(X)\iota(\xi)+\epsilon (X)\iota(\xi)-\iota(X) \epsilon (\xi) \big)-(a_0+b_0) \xi(X).
\end{align}

(1).  By Lemma (3.6) and the relation of the Clifford action and $ {\rm{Tr}}(AB)= {\rm{Tr}}(BA) $, then
   \begin{align}
 {\rm{Tr}}\Big(\widetilde{c}(u)  \widetilde{c}(\xi)\big( \bar{c}(\xi)c(X)+c(X)\widetilde{c}(\xi) \big)\Big)
  = a_{0} b_{0} (a_{0}+b_{0})\xi(X) \xi(u) {\rm{Tr}}(\rm{Id}).
\end{align}

(2). We note that
  \begin{align}
\big(\bar{c}(\xi)c(X)+c(X)\widetilde{c}(\xi)\big) \epsilon (\xi)\iota(\xi)=(a_0-b_0) |\xi|^{2}c(X)\iota(\xi)-(a_0+b_0)  \xi(X)\epsilon (\xi)\iota(\xi).
\end{align}
Then
   \begin{align}
&{\rm{Tr}}\Big(\widetilde{c}(u)  \widetilde{c}(\xi)\big( \bar{c}(\xi)c(X) \epsilon (\xi)\iota(\xi)+c(X)\widetilde{c}(\xi) \epsilon (\xi)\iota(\xi) \big)\Big)\nonumber\\
=& (a_0-b_0) |\xi|^{2}{\rm{Tr}}\Big(\widetilde{c}(u)  \widetilde{c}(\xi)c(X)\iota(\xi)\Big)
-(a_0+b_0)  \xi(X) {\rm{Tr}}\Big(\widetilde{c}(u)  \widetilde{c}(\xi)\epsilon (\xi)\iota(\xi)\Big)\nonumber\\
 =&\big(\frac{a_{0}(3b_{0}^{2} - a_{0}^{2})}{2} |\xi|^{2}\xi(X)\xi(u)+\frac{a_{0}(a_{0}^{2} - b_{0}^{2})}{2} |\xi|^{2}g(u,v)\big){\rm{Tr}}(\rm{Id}).
\end{align}

(3). It is obtained through direct calculation
  \begin{align}
\epsilon (\xi)\iota(\xi)\big(\bar{c}(\xi)c(X)+c(X)\widetilde{c}(\xi)\big)=(a_0-b_0)|\xi|^{2}\big(\xi(X)+  c(X)\epsilon (\xi) \big)-(a_0+b_0)  \xi(X)\epsilon (\xi)\iota(\xi).
\end{align}
Then
 \begin{align}
  &{\rm{Tr}}\Big(\widetilde{c}(u)  \widetilde{c}(\xi)\big( \epsilon (\xi)\iota(\xi)\bar{c}(\xi)c(X)+ \epsilon (\xi)\iota(\xi)c(X)\widetilde{c}(\xi) \big)\Big)\nonumber\\
  =& {\rm{Tr}}\Big(\widetilde{c}(u)  \widetilde{c}(\xi)
  \big( (a_0-b_0)|\xi|^{2}\big(\xi(X)+  c(X)\epsilon (\xi) \big)-(a_0+b_0)  \xi(X)\epsilon (\xi)\iota(\xi)
   \big)\Big)\nonumber\\
 =& (a_0-b_0)|\xi|^{2}\xi(X){\rm{Tr}}\Big(\widetilde{c}(u)  \widetilde{c}(\xi)
  \Big)
  +(a_0-b_0)|\xi|^{2}  {\rm{Tr}}\Big(\widetilde{c}(u)  \widetilde{c}(\xi)
 c(X)\epsilon (\xi) \Big)\nonumber\\
   &  -(a_0+b_0) \xi(X)  {\rm{Tr}}\Big(\widetilde{c}(u) \widetilde{c}(\xi)
\epsilon (\xi)\iota(\xi)
 \Big)\nonumber\\
 =&\big(\frac{a_{0}^{3}b_{0}-a_{0}b_{0}^{3}+2a_{0}b_{0}^{2} }{2}|\xi|^{2}\xi(X) \xi(u)
 - \frac{b_{0}(a_{0}^{2} - b_{0}^{2})}{4} |\xi|^{2}g(u,v)\big){\rm{Tr}}(\rm{Id}).
\end{align}

(4). We note that
  \begin{align}
\epsilon (\xi)\iota(\xi)\big(\bar{c}(\xi)c(X)+c(X)\widetilde{c}(\xi)\big) \epsilon (\xi)\iota(\xi)= -2b_0  |\xi|^{2}\xi(X)\epsilon (\xi)\iota(\xi).
\end{align}
 Then
 \begin{align}
 &{\rm{Tr}}\Big(\widetilde{c}(u)  \widetilde{c}(v) \widetilde{c}(w)\widetilde{c}(\xi)\big( \epsilon (\xi)\iota(\xi)\bar{c}(\xi)c(X)\epsilon (\xi)\iota(\xi)+ \epsilon (\xi)\iota(\xi)c(X)\widetilde{c}(\xi)\epsilon (\xi)\iota(\xi) \big)\Big)\nonumber\\
 = &{\rm{Tr}}\Big(\widetilde{c}(u)  \widetilde{c}(\xi)\big( -2b_0  |\xi|^{2}\xi(X)\epsilon (\xi)\iota(\xi)\big)\Big)\nonumber\\
  = &-2b_0  |\xi|^{2}\xi(X){\rm{Tr}}\Big(\widetilde{c}(u) \widetilde{c}(\xi) \epsilon (\xi)\iota(\xi) \Big)\nonumber\\
  =&a_{0} b_{0}^{2 }(a_{0}+b_{0})|\xi|^{4}\xi(X) \xi(u) {\rm{Tr}}(\rm{Id}).
\end{align}
 \end{proof}
By integrating formula we get
\begin{align}
& \int_{\|\xi\|=1} \operatorname{Tr} [L_{1}  (x_{0} )+L_{2}  (x_{0} ) ](x, \xi)  \sigma(\xi)  \nonumber\\
=&\int_{\|\xi\|=1} \operatorname{Tr} [\widetilde{c}(u)  \sqrt{-1}c(X)
\frac{b_{0}^{4}|\xi|^{2}+(a_{0}^{4}-b_{0}^{4})\varepsilon(\xi)\iota(\xi)}{a_{0}^{4}b_{0}^{4}|\xi|^{6}}  (x_{0} ) ](x, \xi)  \sigma(\xi)  \nonumber\\
&+\int_{\|\xi\|=1} \operatorname{Tr} \Big[ -\sqrt{-1}\widetilde{c}(u)  \widetilde{c}( \xi)
\sigma_{-2}(\widetilde{\Delta}^{-1})\sigma_{1}((\widetilde{D}^{*}\widetilde{D})^{-1})\sigma_{-4}(\widetilde{\Delta}^{-2})  (x_{0} ) \Big](x, \xi)  \sigma(\xi)  \nonumber\\
&+\int_{\|\xi\|=1} \operatorname{Tr} \Big[  -\sqrt{-1}\widetilde{c}(u)
\widetilde{c}(\xi)\sigma_{-4}(\widetilde{\Delta}^{-2})\sigma_{1}((\widetilde{D}^{*}\widetilde{D})^{-1})\sigma_{-2}(\widetilde{\Delta}^{-1}) (x_{0} ) \Big](x, \xi)  \sigma(\xi)  \nonumber\\
=&\frac{\sqrt{-1}(a_{0}^{2}-b_{0}^{2})(a_{0}-b_{0})
(11a_{0}^{4}+16a_{0}^{3}b_{0}+4a_{0}^{2}b_{0}^{3}+11a_{0}^{2}b_{0}^{2}+4a_{0} b_{0}^{4}+8a_{0} b_{0}^{3} -8b_{0}^{4} )}{16a_{0}^{6}b_{0}^{4}}
 g(u,X) {\rm{Tr}}(\rm{Id}).
 \end{align}
Then we obtain the spectral form associated with nonminimal de Rham-Hodge operators.
 \begin{thm}\label{mrthm}
 	Let $M$ be an four-dimensional oriented compact spin Riemannian manifold,  with the trilinear Clifford multiplication by functional of differential one-form $\widetilde{c}(u)$, the spectral torsion $\mathscr{T}_{\widetilde{D}}$ for  nonminimal de Rham-Hodge operators  $\widetilde{D}=a_0\text{d}+b_0\delta+\sqrt{-1}c(X),  \widetilde{D}^{*}=b_0\text{d}+a_0\delta+\sqrt{-1}c(X)$  equals to
\begin{align}
 	&\mathscr{T}_{1}\big(\widetilde{c}(u) \big)
= \mathrm{Wres}\big(\widetilde{c}(u)  \widetilde{D}(\widetilde{D}^{*}\widetilde{D})^{-2}\big) \nonumber\\
=& \int_{M}\bigg\{
\frac{\sqrt{-1}(a_{0}^{2}-b_{0}^{2})(a_{0}-b_{0})
(11a_{0}^{4}+16a_{0}^{3}b_{0}+4a_{0}^{2}b_{0}^{3}+11a_{0}^{2}b_{0}^{2}+4a_{0} b_{0}^{4}+8a_{0} b_{0}^{3} -8b_{0}^{4} )}{4a_{0}^{6}b_{0}^{4}}
 g(u,X) {\rm Vol}(S^{3})
\bigg\}{\rm d Vol}_M.
\end{align}
 \end{thm}

\section{The spectral torsion for nonminimal de-Rham Hodge operator on four-dimensional manifold with boundary }
The purpose of this section is to specify the spectral torsion on manifold with boundary  for the nonminimal de Rham-Hodge operators
$\widetilde{D}=a_0\text{d}+b_0\delta+\sqrt{-1}c(X)$ and $\widetilde{D}^{*}=b_0\text{d}+a_0\delta+\sqrt{-1}c(X)$.
 Let $M$ be a compact oriented Riemannian manifold of even dimension $n=4$.
 An application of (2.1.4) in \cite{Wa1} shows that
\begin{align}
\widetilde{{\rm Wres}}\Big[\pi^+\Big(\widetilde{c}(u)\widetilde{c}(v)\widetilde{c}(w)\widetilde{D}(\widetilde{D}^{*}\widetilde{D})^{-1}\Big)
\circ\pi^+(\widetilde{D}^{*}\widetilde{D})^{-1}\Big]
=&\mathrm{Wres}\big(\widetilde{c}(u) \widetilde{c}(v) \widetilde{c}(w) \widetilde{D}(\widetilde{D}^{*}\widetilde{D})^{-2}\big)+\int_{\partial M}\Psi,
\end{align}
where
\begin{align}
\Psi_{1}=&\int_{|\xi'|=1}\int^{+\infty}_{-\infty}
{\rm Tr}_{\wedge^{*}(T^{*}M)}
\big[ \sigma^+_{-1}(\widetilde{c}(u) \widetilde{c}(v) \widetilde{c}(w)\widetilde{D}(\widetilde{D}^{*}\widetilde{D})^{-1})(x',0,\xi',\xi_n)\nonumber\\
&\times
\partial_{\xi_n}\sigma_{-2} (\widetilde{D}^{*}\widetilde{D})^{-1}(x',0,\xi',\xi_n)\big]\rm{d}\xi \sigma(\xi')\rm{d}x',
\end{align}
and $r-k+|\alpha|+\ell-j-1=-4,r\leq-1,\ell\leq-2$.

An easy calculation gives
  \begin{align}
 &\pi^+_{\xi_n}\big( \sigma_{-1}(\widetilde{c}(u)\widetilde{c}(v)\widetilde{c}(w)\widetilde{D}(\widetilde{D}^{*}\widetilde{D})^{-1})\big)\nonumber\\
 =&\widetilde{c}(u)\widetilde{c}(v)\widetilde{c}(w)\pi^+_{\xi_n}\big( \sigma_{-1} (\widetilde{D}(\widetilde{D}^{*}\widetilde{D})^{-1})\big)\nonumber\\
  =& \widetilde{c}(u)\widetilde{c}(v)\widetilde{c}(w)\pi^+_{\xi_n}\big( \sigma_{-1} ( (\widetilde{D}^{*} )^{-1})\big)\nonumber\\
=&\widetilde{c}(u)\widetilde{c}(v)\widetilde{c}(w)
\Big(\frac{ \epsilon (\xi')+\sqrt{-1} \epsilon (dx_{n} )}{2a_{0}(\xi_{n}-i)} -\frac{\iota (\xi')+\sqrt{-1} \iota(dx_{n})}{2b_{0}(\xi_{n}-i)}\Big),
 \end{align}
where some basic facts and formulae about Boutet de Monvel's calculus which can be found  in Sec.2 in \cite{Wa1}.
By (2.16), we get
  \begin{align}
&\partial_{\xi_n}\big(\sigma_{-2}((\widetilde{D}^{*}\widetilde{D})^{-1})\big)   (x_0)|_{|\xi'|=1}\nonumber\\
=&\partial_{\xi_n}\big(\frac{b_{0}^{2}|\xi|^{2}+(a_{0}^{2}-b_{0}^{2})\varepsilon(\xi)\iota(\xi)}{a_{0}^{2}b_{0}^{2}|\xi|^{4}}\big) (x_0)|_{|\xi'|=1}\nonumber\\
=& \frac{1}{a_{0}^{2}}\frac{-2\xi_n}{ (1+\xi_n^2)^{2}}
+ \frac{a_{0}^{2}-b_{0}^{2}}{a_{0}^{2}b_{0}^{2}}
\Big(\frac{-4\xi_n}{ (1+\xi_n^2)^{3}}\epsilon (dx_{n})\iota(dx_{n})+\frac{-2\xi_n^{2}+2\xi_n}{ (1+\xi_n^2)^{3}} \epsilon (\xi')\iota (\xi')\nonumber\\
&+\frac{-3\xi_n^{2}+1}{ (1+\xi_n^2)^{3}} \epsilon (\xi')\iota(dx_{n})+\frac{-3\xi_n^{2}+1}{ (1+\xi_n^2)^{3}}\epsilon (dx_{n} )\iota (\xi')  \Big).
\end{align}
By Lemma (3.6) and the relation of the Clifford action and $ {\rm{Tr}}(AB)= {\rm{Tr}}(BA) $, we have the equality:
 \begin{align}
 &{\rm{Tr}}\big(\widetilde{c}(u) \widetilde{c}(v) \widetilde{c}(w)\epsilon (\xi')\epsilon (dx_{n} )\iota (\xi') \big)\nonumber\\
 =& - {\rm{Tr}}\big(\widetilde{c}(u) \widetilde{c}(v) \widetilde{c}(w)\epsilon (\xi')\iota (\xi')\epsilon (dx_{n} ) \big) \nonumber\\
  =& - {\rm{Tr}}\big(\widetilde{c}(u) \widetilde{c}(v) \widetilde{c}(w)
  \big(- \iota (\xi') \epsilon (\xi') +  |\xi'|^{2 }     \big)\epsilon (dx_{n} ) \big) \nonumber\\
    =& {\rm{Tr}}\big(\widetilde{c}(u) \widetilde{c}(v) \widetilde{c}(w)
\iota (\xi') \epsilon (\xi')\epsilon (dx_{n} ) \big)
   - |\xi'|^{2 } {\rm{Tr}}\big(\widetilde{c}(u) \widetilde{c}(v) \widetilde{c}(w)\epsilon (dx_{n} ) \big) \nonumber\\
  =&\cdots \nonumber\\
  =&-{\rm{Tr}}\big(\widetilde{c}(u) \widetilde{c}(v) \widetilde{c}(w)\epsilon (\xi')\epsilon (dx_{n} )\iota (\xi') \big)- |\xi'|^{2 } {\rm{Tr}}\big(\widetilde{c}(u) \widetilde{c}(v) \widetilde{c}(w)\epsilon (dx_{n} ) \big)\nonumber\\
  &+a_{0}\xi'(u){\rm{Tr}}\big(\widetilde{c}(v)  \widetilde{c}(w)\epsilon (\xi')\epsilon (dx_{n} ) \big)
  -a_{0}\xi'(v){\rm{Tr}}\big(\widetilde{c}(u)  \widetilde{c}(w)\epsilon (\xi')\epsilon (dx_{n} ) \big)\nonumber\\
    &+a_{0}\xi'(w){\rm{Tr}}\big(\widetilde{c}(u)  \widetilde{c}(v)\epsilon (\xi')\epsilon (dx_{n} ) \big)
\end{align}
By transferring the first term on the right side of (5.5) to the left, we obtain
 \begin{align}
 &{\rm{Tr}}\big(\widetilde{c}(u) \widetilde{c}(v) \widetilde{c}(w)\epsilon (\xi')\epsilon (dx_{n} )\iota (\xi') \big)\nonumber\\
  =&  \frac{1}{2}\Big(
  - |\xi'|^{2 } {\rm{Tr}}\big(\widetilde{c}(u) \widetilde{c}(v) \widetilde{c}(w)\epsilon (dx_{n} ) \big)
+a_{0}\xi'(u){\rm{Tr}}\big(\widetilde{c}(v)  \widetilde{c}(w)\epsilon (\xi')\epsilon (dx_{n} ) \big)\nonumber\\
  &-a_{0}\xi'(v){\rm{Tr}}\big(\widetilde{c}(u)  \widetilde{c}(w)\epsilon (\xi')\epsilon (dx_{n} ) \big)
     +a_{0}\xi'(w){\rm{Tr}}\big(\widetilde{c}(u)  \widetilde{c}(v)\epsilon (\xi')\epsilon (dx_{n} ) \big)\Big)\nonumber\\
 =& -\frac{a_{0}b_{0}^{2}}{4}  |\xi'|^{2 }\big( u_{n}g(v,w)-v_{n}g(u,w)+w_{n}g(u,v) \big).
\end{align}
In the same way we have
 \begin{align}
 &{\rm{Tr}}\big(\widetilde{c}(u) \widetilde{c}(v) \widetilde{c}(w)\iota (\xi')\epsilon (\xi')\iota (dx_{n} ) \big)\nonumber\\
  =&  \frac{1}{2}\Big(
a_{0}u_{n}{\rm{Tr}}\big(  \widetilde{c}(v) \widetilde{c}(w)\iota (\xi')\epsilon (\xi') \big)
   -a_{0}v_{n}{\rm{Tr}}\big(  \widetilde{c}(u) \widetilde{c}(w)\iota (\xi')\epsilon (\xi') \big)
  +a_{0}w_{n}{\rm{Tr}}\big(  \widetilde{c}(u) \widetilde{c}(v)\iota (\xi')\epsilon (\xi') \big)\Big)\nonumber\\
 =& -\frac{a_{0}^{2}b_{0}}{4}  |\xi'|^{2 }\big( u_{n}g(v,w)-v_{n}g(u,w)+w_{n}g(u,v) \big).
\end{align}
Similarly, we obtain
\begin{align}
 &{\rm{Tr}}\big(\widetilde{c}(u) \widetilde{c}(v)\widetilde{c}(w) \epsilon (dx_{n}) \big)
 =\frac{a_{0} b_{0}^{2}}{2}\big( u_{n}g(v,w)-v_{n}g(u,w)+w_{n}g(u,v) \big){\rm{Tr}}(\rm{Id});\nonumber\\
&{\rm{Tr}}\big(\widetilde{c}(u) \widetilde{c}(v)\widetilde{c}(w) \iota  (dx_{n}) \big)
 =-\frac{a_{0}^{2} b_{0}}{2}\big( u_{n}g(v,w)-v_{n}g(u,w)+w_{n}g(u,v) \big){\rm{Tr}}(\rm{Id});\nonumber\\
  &{\rm{Tr}}\big(\widetilde{c}(u) \widetilde{c}(v)\widetilde{c}(w) \iota  (dx_{n})  \epsilon (dx_{n}) \iota  (dx_{n}) \big)
 =-\frac{a_{0}^{2} b_{0}}{2}|dx_{n}|^{2}\big( u_{n}g(v,w)-v_{n}g(u,w)+w_{n}g(u,v) \big){\rm{Tr}}(\rm{Id});\nonumber\\
  &{\rm{Tr}}\big(\widetilde{c}(u) \widetilde{c}(v)\widetilde{c}(w) \iota  (dx_{n})  \epsilon (\xi')\iota(\xi') \big)
 =-\frac{a_{0}^{2} b_{0}}{4}|\xi'|^{2}\big( u_{n}g(v,w)-v_{n}g(u,w)+w_{n}g(u,v) \big){\rm{Tr}}(\rm{Id}).
\end{align}
From (5.3) and (5.4) we get
\begin{align}
&{\rm Tr}_{\wedge^{*}(T^{*}M)}\big[ \sigma^+_{-1}(c(u)c(v)c(w)\widetilde{D}(\widetilde{D}^{*}\widetilde{D})^{-1})(x',0,\xi',\xi_n)
 \partial_{\xi_n}\sigma_{-2} (\widetilde{D}^{*}\widetilde{D})^{-1}(x',0,\xi',\xi_n)\big]\nonumber\\
=&  \frac{a_{0}^{2}-b_{0}^{2}}{2a_{0}^{3}b_{0}^{2}}\frac{1-3\xi_n^{2}}{ (\xi_n-i)(1+\xi_n^2)^{3}}{\rm{Tr}}\big(\widetilde{c}(u) \widetilde{c}(v) \widetilde{c}(w)\epsilon (\xi')\epsilon (dx_{n} )\iota (\xi') \big)\nonumber\\
&- \frac{a_{0}^{2}-b_{0}^{2}}{2a_{0}^{2}b_{0}^{3}}\frac{1-3\xi_n^{2}}{ (\xi_n-i)(1+\xi_n^2)^{3}}{\rm{Tr}}\big(\widetilde{c}(u) \widetilde{c}(v) \widetilde{c}(w)\iota (\xi')\epsilon (\xi')\epsilon (dx_{n} ) \big)\nonumber\\
&- \frac{1}{a_{0}^{3} }\frac{i\xi_n}{ (\xi_n-i)(1+\xi_n^2)^{2}}{\rm{Tr}}\big(\widetilde{c}(u) \widetilde{c}(v) \widetilde{c}(w) \epsilon (dx_{n} ) \big)\nonumber\\
&+\frac{a_{0}^{2}-b_{0}^{2}}{a_{0}^{3}b_{0}^{3}}\frac{-i\xi_n^{2}+i\xi_n}{ (\xi_n-i)(1+\xi_n^2)^{3}}{\rm{Tr}}\big(\widetilde{c}(u) \widetilde{c}(v) \widetilde{c}(w)\epsilon (dx_{n} )\epsilon (\xi') \iota (\xi')\big)\nonumber\\
&+ \frac{1}{a_{0}^{2}b_{0}  }\frac{i\xi_n}{ (\xi_n-i)(1+\xi_n^2)^{2}}{\rm{Tr}}\big(\widetilde{c}(u) \widetilde{c}(v) \widetilde{c}(w)  \iota(dx_{n} ) \big)\nonumber\\
&+ \frac{a_{0}^{2}-b_{0}^{2}}{a_{0}^{2}b_{0}^{3}}\frac{2i\xi_n}{ (\xi_n-i)(1+\xi_n^2)^{3}}{\rm{Tr}}\big(\widetilde{c}(u) \widetilde{c}(v) \widetilde{c}(w)\iota(dx_{n})\epsilon(dx_{n})\iota(dx_{n})\big)\nonumber\\
&+ \frac{a_{0}^{2}-b_{0}^{2}}{a_{0}^{2}b_{0}^{3}}\frac{i\xi_n^{2}-i\xi_n}{ (\xi_n-i)(1+\xi_n^2)^{3}}{\rm{Tr}}\big(\widetilde{c}(u) \widetilde{c}(v) \widetilde{c}(w)\iota(dx_{n})\epsilon(\xi')\iota(\xi')\big).
\end{align}
Substituting (5.9) into (5.2) we get
\begin{align}
 \Psi=&\int_{|\xi'|=1}\int^{+\infty}_{-\infty}
{\rm Tr}_{\wedge^{*}(T^{*}M)}
\big[ \sigma^+_{-1}(\widetilde{c}(u) \widetilde{c}(v) \widetilde{c}(w)\widetilde{D}(\widetilde{D}^{*}\widetilde{D})^{-1})(x',0,\xi',\xi_n)\nonumber\\
&\times
\partial_{\xi_n}\sigma_{-2} (\widetilde{D}^{*}\widetilde{D})^{-1}(x',0,\xi',\xi_n)\big]\rm{d}\xi \sigma(\xi')\rm{d}x'\nonumber\\
  =& \frac{a_{0}^{4}-b_{0}^{4}-\sqrt{-1}(2a_{0}^{4}+2a_{0}^{2}b_{0}^{2}-4a_{0}^{4} )}{64a_{0}^{2}b_{0}^{2}}
 \pi(u_{n}g(v,w)-v_{n}g(u,w)+w_{n}g(u,v)){\rm Tr}_{\wedge^{*}(T^{*}M)}[ {\rm Id}] {\rm{vol}}(S^{2})
{\rm d}\rm{vol}_{\partial_{M}}.
\end{align}

Summing up (5.10) and  Theorem 3.9 leads to the spectral torsion $\mathscr{\widetilde{T}}_{\widetilde{D}}$ for manifold with boundary  as follows.
\begin{thm}
With the trilinear Clifford multiplication by functional of differential one-forms
$\widetilde{c}(u), \widetilde{c}(v), \widetilde{c}(w)$, the spectral torsion for  for nonminimal de-Rham Hodge operator on four-dimensional manifold with boundary  equals to
\begin{align}
&\mathscr{\widetilde{T}}_{\widetilde{D}}(\widetilde{c}(u), \widetilde{c}(v), \widetilde{c}(w))\nonumber\\
=& \int_{M}\bigg\{
\frac{12\sqrt{-1}(a_{0}^{4}-b_{0}^{4})(2b_{0}^{2}-3a_{0}^{2}-a_{0}b_{0})}{16a_{0}^{4}b_{0}^{3}}
\Big[g(u,X)g(v,w)-g(v,X)g(u,w)+g(w,X)g(u,v)\Big] {\rm Vol}(S^{3})
\bigg\}{\rm d Vol}_M  \nonumber\\
&+\int_{\partial_{M}}\frac{a_{0}^{4}-b_{0}^{4}-\sqrt{-1}(2a_{0}^{4}+2a_{0}^{2}b_{0}^{2}-4a_{0}^{4} )}{16a_{0}^{2}b_{0}^{2}}
 \pi(u_{n}g(v,w)-v_{n}g(u,w)+w_{n}g(u,v)) {\rm Vol}(S^{2})
{\rm d}\rm{vol}_{\partial_{M}}.
\end{align}
\end{thm}

\section*{ Acknowledgements}
~The work of the  first  author was supported by NSFC.11501414.
The work of the second author was supported by NSFC. 11771070. The authors also thank the referee
for his (or her) careful reading and helpful comments.

\end{document}